%% file: arxiv2.tex
\documentclass[11pt]{article}
\usepackage{amsmath,amsfonts,amsthm,amssymb}
\usepackage{graphics,color,url}
\usepackage{graphicx}
\usepackage{epsfig}
\usepackage{fullpage}

\newtheorem{theorem}{Theorem}[section]

\newtheorem{lemma}[theorem]{Lemma}
\newtheorem{claim}[theorem]{Claim}

\newtheorem{definition}[theorem]{Definition}
\newtheorem{coro}[theorem]{Corollary}

\newcommand{\mainmechanism}{{\sc Pruning-Lifting Mechanism}}
\newcommand{\pathmechanism}{{\sc Pruning-Lifting $k$-Paths Mechanism}}
\newcommand{\calM}{{\cal M}}
\newcommand{\calE}{{\cal E}}
\newcommand{\calF}{{\cal F}}
\newcommand{\calH}{{\cal H}}
\newcommand{\calS}{{\cal S}}
\newcommand{\calG}{{\cal G}}
\newcommand{\CL}{{\cal CL}}
\newcommand{\vecc}{{\mathbf c}}
\newcommand{\vecb}{{\mathbf b}}
\newcommand{\vecw}{{\mathbf w}}
\newcommand{\vecy}{{\mathbf y}}

\title{\textbf{Frugal Mechanism Design via Spectral Techniques}}

\author
{\\Ning Chen\footnote{Division of Mathematical Sciences,
School of Physical and Mathematical Sciences,
Nanyang Technological University, Singapore.
Email: {\tt ningc@ntu.edu.sg, eelkind@ntu.edu.sg, ngravin@pmail.ntu.edu.sg.}}
\and \\Edith Elkind\footnotemark[1]
\and \\Nick Gravin\footnotemark[1] \footnotemark[2]
\and \\Fedor Petrov\footnote{St.Petersburg Department of Steklov Mathematical
Institute RAS, Russia.}
}

\date{}

\begin{document}

\maketitle
\thispagestyle{empty}

\begin{abstract}
We study the design of truthful mechanisms for set systems, i.e.,
scenarios where a customer needs to
hire a team of agents to perform a complex task. In this setting,
frugality~\cite{archer} provides a measure to evaluate
the ``cost of truthfulness'', that is, the overpayment of
a truthful mechanism relative to the ``fair'' payment.

We propose a uniform scheme for
designing frugal truthful mechanisms for general set systems.
Our scheme is based on scaling the agents' bids using the eigenvector of
a matrix that encodes the interdependencies between the agents.
We demonstrate that the $r$-out-of-$k$-system mechanism and
the $^{\sqrt{\ }}$-mechanism for buying a path in a graph~\cite{anna}
can be viewed as instantiations of our scheme. We then apply our scheme
to two other classes of set systems, namely,
vertex cover systems~\cite{calinescu,edith3} and
$k$-path systems, in which a
customer needs to purchase $k$ edge-disjoint source-sink paths.
For both settings, we bound the frugality of our mechanism in terms of the
largest eigenvalue of the respective interdependency matrix.

We show that our mechanism is optimal for a large subclass of vertex cover systems
satisfying a simple local sparsity condition. For $k$-path systems, while our mechanism is within a factor of $k+1$
from optimal, we show that it is, in fact, {\em optimal}, when one uses
a modified definition of frugality proposed in~\cite{edith3}.
Our lower bound argument combines spectral
techniques and Young's inequality, and is applicable to
all set systems. As both $r$-out-of-$k$ systems and single path systems
can be viewed as special cases of $k$-path systems,
our result improves the lower bounds of~\cite{anna}
and answers several open questions proposed in~\cite{anna}.
\end{abstract}

\newpage

\section{Introduction}

Consider a scenario where a customer wishes to purchase the rights to have data routed on his
behalf from a source $s$ to a destination $t$ in a network where each edge
is owned by a selfishly motivated agent. Each agent incurs a privately known
cost if the data is routed through his edge, and wants to be compensated
for this cost, and, if possible, make a profit. The customer
needs to decide which edges to buy, and wants to minimize his total expense.

This problem is a special case of the \textit{hiring-a-team}
problem~\cite{talwar,anna,yokoo,ning,edith3}: Given a set of agents $\calE$,
a customer wishes to hire a team of agents capable of performing a certain complex task
on his behalf. A subset $S\subseteq \calE$ is said to be \textit{feasible}
if the  agents in $S$ can jointly perform the complex task.
This scenario can be described by a \textit{set system} $(\calE, \calF)$,
where $\calE$ is the set of agents and $\calF$ is the collection of feasible sets.
Each agent $e\in \calE$ can perform a simple task at a privately known cost $c(e)$.
In such environments, a natural way to make the hiring decisions is by
means of \textit{mechanisms} ---
Each agent $e$ submits a \textit{bid} $b(e)$, i.e., the payment that he wants to
receive, and based on these bids
the customer selects a feasible set $S\in\calF$ (the set of \textit{winners}),
and determines the payment to each agent in $S$.

A desirable property of mechanisms is that of \textit{truthfulness}:
It should be in the best interest of every agent $e$ to bid his true cost, i.e.
to set $b(e)=c(e)$ no matter what bids other agents submit;
that is, truth-telling should be a dominant strategy for every agent.
Truthfulness is a strong and very appealing concept: it obviates the need for agents to
perform complex strategic computations, even if they do not know the costs and strategies
of others. This property is especially important in the Internet and electronic commerce
settings, as most protocols are executed instantly.

One of the most celebrated truthful designs is the VCG mechanism~\cite{vickrey,clarke,groves},
where the feasible set with the smallest total bid wins, and the payment to each
agent $e$ in the winning set is his threshold bid, i.e., the highest value
that $e$ could have bid to still be part of a winning set.
While VCG mechanism is truthful, on the negative side,
it can make the customer pay far more than the true cost of the winning set,
or even the cheapest alternative, as illustrated by the following example:
There are two parallel paths $P_1$ and $P_2$ from $s$ to $t$, $P_1$ has one edge with
cost 1 and $P_2$ has $n$ edges with cost 0 each.
VCG selects $P_2$ as the winning path and pays 1 to every edge in $P_2$.
Hence, the total payment of VCG is $n$, the number of edges in $P_2$, which is far more
than the total cost of both $P_1$ and $P_2$.

The VCG overpayment property illustrated above is clearly undesirable
from the customer's perspective, and thus
motivates the search for truthful mechanisms that are \textit{frugal},
i.e., select a feasible set and induce truthful cost revelation without
resulting in high overpayment. However, formalizing the notion of frugality
is a challenging problem, as it is not immediately clear what the payment of
a mechanism should be compared to. A natural candidate for this benchmark
is the total cost of the closest competitor, i.e., the cost of the cheapest
feasible set among those that are disjoint from the winning set.
This definition coincides with the second highest bid in
single-item auctions and has been used in, e.g.,~\cite{archer1,archer,talwar,edith}.
However, as observed by Karlin, Kempe and Tamir~\cite{anna}, such
feasible set may not exist at all, even in monopoly-free set systems
(i.e., set systems where no agent appears in all feasible sets). To deal with this
problem, \cite{anna} proposed an alternative benchmark, which is bounded
for any monopoly-free set system and is closely related to the buyer-optimal
Nash equilibrium of first-price auctions (see Definition~\ref{def-benchmark}).
Nash equilibrium corresponds to a stable outcome of the bargaining process, and
therefore provides a natural lower bound on the total payment of any
dominant strategy mechanism. Throughout the paper,
we use the benchmark of~\cite{anna}, as well as its somewhat more relaxed
variant suggested in~\cite{edith3} to study frugality of truthful mechanisms.

\subsection{Our Results}

\paragraph{Uniform Frugal Truthful Mechanisms.} We propose a uniform scheme,
which we call \mainmechanism,
to design frugal truthful mechanisms for set systems. At a high-level
view, this mechanism consists of two key steps: pruning and lifting.
\begin{itemize}
\item Pruning. In a general set system, the relationships among the agents
can be arbitrarily complicated. Thus, in the pruning step, we remove agents from the system
so as to expose the structure of the competition. Intuitively, the goal is to keep
only the agents who are going to play a role in determining the bids in Nash equilibrium;
this enables us to compare the payoffs of our mechanism to the total equilibrium payment.
Since we decide which agents to prune based on their bids, we have to make
our choices carefully so as to preserve truthfulness.

\item Lifting. The goal of the lifting process is to ``lift" the bid of each remaining agent
so as to take into account the size of each feasible set. For this purpose, we use
a graph-theoretic approach inspired by the ideas in~\cite{anna}.
Namely, we construct a graph $\calH$ whose vertices are agents, and there is an edge between
two agents $e$ and $e'$ if removing both $e$ and $e'$ results in a system
with no feasible solution. We call $\calH$ the dependency graph of the pruned system.
We then compute the largest eigenvalue of $\calH$
(or, more precisely, the maximum of the largest eigenvalues of its connected components),
which we denote by $\alpha_\calH$, and scale the bid of each agent by the respective
coordinate of the eigenvector that corresponds to $\alpha_\calH$.
\end{itemize}
A given set system may be pruned in different ways, thus leading to different
values of $\alpha_\calH$.
We will refer to the largest of them, i.e., $\alpha=\sup_{\calH}\alpha_\calH$,
as the {\em eigenvalue} of our set system. It turns out that this quantity plays an important
role in our analysis.

We show that the $r$-out-of-$k$-system mechanism and the $^{\sqrt{\ }}$-mechanism
for the single path problem that were presented in~\cite{anna}
can be viewed as instantiations of our \mainmechanism. We then apply our scheme to two other
classes of set systems:  vertex cover systems, where the goal is to buy a vertex cover in a given graph,
and $k$-path systems, where the goal is to buy $k$ edge-disjoint paths between two vertices of a
given graph.

The $k$-path problem generalizes both the $r$-out-of-$k$ problem
and the single path problem, and captures many other natural scenarios.
However, this problem received limited attention from the algorithmic mechanism
design community so far (see, however,~\cite{karger}),
perhaps due to its inherent difficulty: the interactions
among the agents can be quite complex, and, prior to this work, it was not known
how to characterize Nash equilibria of the first-price auctions
for this setting in terms of the network structure.
In this paper, we obtain a strong lower bound on the total payments in Nash equilibria.
We then use this bound to show that a natural variant of the \mainmechanism\
that prunes all edges except those in the cheapest flow of size $k+1$ has frugality
ratio $\alpha\frac{k+1}{k}$.
Moreover, we show that this bound can be improved by a factor of $k+1$
if we consider a weaker payment bound suggested in~\cite{edith3}, which corresponds
to a buyer-pessimal rather than buyer-optimal Nash equilibrium
(i.e., the difference between two frugality bounds is akin to that between
the price of anarchy and the price of stability).

For the vertex cover problem, an earlier paper~\cite{edith3}
described a mechanism with frugality ratio $2\Delta$,
where $\Delta$ is the maximum degree of the input graph.
Our approach results in a mechanism whose frugality
ratio equals to the largest eigenvalue $\alpha$ of the adjacency matrix of the input graph.
As $\alpha\le \Delta$ for any graph $G$, this means that we improve the
result of~\cite{edith3} by at least a factor of 2 for all graphs.
Surprisingly, this stronger bound can be obtained
by a simple modification of the analysis in~\cite{edith3}.

\paragraph{Lower Bounds.}
We complement the bounds on the frugality of the \mainmechanism\ by proving
strong lower bounds on the frugality of (almost) any truthful mechanism.
In more detail, we exhibit a family of cost vectors on which
the payment of any {\em measurable} truthful mechanism
can be lower-bounded in terms of $\alpha$, where we call a
mechanism measurable if the payment to any
agent --- as a function of other agents'
bids --- is a Lebesgue measurable
function. Lebesgue measurability is a much weaker condition
than continuity or monotonicity; indeed, a mechanism that does not satisfy
this condition is unlikely to be practically implementable!
Our argument relies on Young's inequality and applies to any set system.

To turn this lower bound on payments into a lower bound on frugality,
we need to understand the structure of Nash equilibria for the bid vectors
employed in our proof. For $k$-path systems, we can achieve this by using
our characterization of Nash equilibria in such systems. As a result,
we obtain a lower bound on frugality of any ``measurable'' truthful mechanism
that shows that our mechanism is within a factor of $(k+1)$ from optimal.
Moreover, it is, in fact, optimal, with respect to the weaker payment
bound of~\cite{edith3}.
For $r$-out-of-$k$ systems and
single path systems, our bound improves the lower bounds on frugality
given in~\cite{anna} by a factor of $2$ and $\sqrt{2}$, respectively. Our results give
strong evidence that simply choosing the cheapest $(k+1)$-flow mechanism,
which generalizes the $^{\sqrt{\ }}$-mechanism~\cite{anna} for $k=1$, is indeed an optimal
frugal mechanism for paths systems.

For the vertex cover problem, characterizing
the Nash equilibria turns out to be a more difficult task: in this case, the graph
$\calH$ is equal to the input graph, and therefore is not guaranteed to have
any regularity properties. However, we can still obtain non-trivial upper bounds
on the payments in Nash equilibria. These bounds enable us to show that our
mechanism for vertex cover is optimal for all triangle-free graphs, and, more
generally, for all graphs that satisfy a simple local sparsity condition.

\subsection{Related Work}\label{section-related-work}

There is a substantial literature on designing mechanisms with small payment
for shortest path systems~\cite{archer,edith,feigenbaum,czumaj,edith2,karger,yan}
as well as for other set systems~\cite{talwar,calinescu,bikhchandani,anna,edith3},
starting with the seminal work of Nisan and Ronen~\cite{nisan}.
Our work is most closely related to~\cite{anna},~\cite{edith3} and~\cite{yan}: we employ
the frugality benchmark defined in~\cite{anna}, improve the bounds of~\cite{anna}
and~\cite{edith3}, and generalize the result of~\cite{yan}.

Simultaneously and independently, the idea of bounding frugality
ratios of set system auctions in terms of eigenvalues of certain matrices
was proposed by Kempe et al.~\cite{kempe}.
In contrast with our work, in~\cite{kempe} the authors
only study the frugality ratio of their mechanisms with respect to the
relaxed payment bound (see Section~\ref{section-multi-path}).
Their approach results in a 2-competitive mechanism for vertex cover systems,
2(k+1)-competitive mechanism for $k$-path systems, and a $4$-competitive mechanism
for cut auctions.

\section{Preliminaries}

A {\em set system} $(\calE, \calF)$ is given by a set $\calE$ of {\em agents} and a collection $\calF\subseteq 2^\calE$
of {\em feasible sets}. We restrict our attention to {\em monopoly-free} set systems,
i.e., we require $\bigcap_{S\in\calF}S=\emptyset$. Each agent $e\in \calE$ has a privately known \textit{cost} $c(e)$ that represents
the expenses that agent $e$ incurs if he is involved in performing the task.

A \textit{mechanism} for a set system $(\calE, \calF)$ takes a {\em bid vector} $\vecb=(b(e))_{e\in\calE}$ as input, where $b(e)\ge c(e)$ for any $e\in \calE$,
and outputs a set of \textit{winners} $S\in \calF$ and
a \textit{payment} $p(e)$ for each $e\in \calE$. We require mechanisms
to satisfy \textit{voluntary participation}, i.e. $p(e)\ge b(e)$ for each $e\in S$ and $p(e)=0$ for each $e\notin S$.

Given the output of a mechanism, the \textit{utility} of an agent $e$ is $p(e)-c(e)$ if $e$ is a winner
and 0 otherwise.
We assume that agents are rational, i.e. aim to maximize their own utility. Thus,
they may lie about their true costs, i.e. bid $b(e)\neq c(e)$
if they can profit by doing so. We say that a mechanism is \textit{truthful} if every agent maximizes his
utility by bidding his true value, no matter what bids other agents
submit. A weaker solution concept is that of {\em Nash equilibrium}:
a bid vector constitutes a (pure) Nash equilibrium if no agent can increase his utility by unilaterally
changing his bid. Nash equilibria describe stable states of the market and can be seen as natural outcomes
of a bargaining process.

There is a well-known characterization of winner selection rules that yield truthful mechanisms.

\begin{theorem}[\cite{krishna,archer}]\label{theorem-monotone}
A mechanism is truthful if and only if its winner selection rule is \textup{monotone},
i.e., no losing agent can become a winner by increasing his bid, given the fixed bids of
all other agents. Further, for a given monotone selection rule, there is a unique truthful
mechanism with this selection rule: the payment to each winner is his {\em threshold bid},
i.e. the highest value he could bid and still win.
\end{theorem}

An example of a truthful set system auction is given by
the VCG mechanism~\cite{vickrey,clarke,groves}.
However, as discussed in Section 1, VCG often
results in a large overpayment to winners.
Another natural mechanism for buying a set is the \textit{first-price
auction}: given the bid vector $\vecb$, pick a subset $S\in \calF$
minimizing $b(S)$, and pay each winner $e\in S$ his bid $b(e)$.
While the first-price auction is not truthful, and more generally,
does not possess dominant strategies, it essentially
admits a Nash equilibrium with a relatively small
total payment. (More accurately, as observed by~\cite{karger}, a
first-price auction may not have a pure strategy Nash
equilibrium. However, this non-existence result
can be circumvented in several ways, e.g. by
considering instead an $\varepsilon$-Nash equilibrium
for arbitrarily small $\varepsilon>0$ or using oracle access to the true
costs of agents to break ties.)
The payment in a buyer-optimal Nash equilibrium would constitute
a natural benchmark for truthful mechanisms. However, due to the difficulties described above,
we use instead the following benchmark proposed by Karlin et al.~\cite{anna},
which captures the main properties of a Nash equilibrium.

\begin{definition}[Benchmark $\nu(\mathbf{c})$~\cite{anna}]\label{def-benchmark}
Given a set system $(\calE, \calF)$, and a feasible set $S\in \calF$ of
minimum total cost w.r.t. $\vecc$,
let $\nu(\mathbf{c})$ be the value of an
optimal solution to the following optimization problem:
\begin{eqnarray*}
\min \ \ & & \sum_{e\in S}b(e) \\
      \text{\em s.t.} \ \ & & (1) \ b(e) \ge c(e) \text{\em \ for all } e\in \calE \\
           & & (2) \ \sum_{e\in S\setminus T} b(e) \le \sum_{e\in T\setminus S} c(e)
                     \text{\em \ for all } \ T\in \calF \\
           & & (3) \ \text{\em For every $e\in S$ there is a $T\in \calF$ s.t. $e\notin T$ and }
                   \sum_{e'\in S\setminus T} b(e') = \sum_{e'\in T\setminus S} c(e')
\end{eqnarray*}
\end{definition}

Intuitively, in the optimal solution of the above system, $S$ is the set of winners in
the first-price auction. By condition (3), no winner $e\in S$ can improve his utility by
increasing his bid $b(e)$, as he would not be a winner anymore. In addition,
by conditions (1) and (2), no agent $e\in\calE\setminus S$ can obtain a positive utility
by decreasing his bid. Hence, $\nu(\mathbf{c})$ gives the value of the cheapest Nash equilibrium of
the first-price auction assuming that the most ``efficient" feasible set $S$ wins.

\begin{definition}[Frugality Ratio]
Let $\mathcal{M}$ be a truthful mechanism for the set system $(\calE, \calF)$ and
let $p_{_\mathcal{M}}(\mathbf{c})$
denote the total payment of $\mathcal{M}$ when the true costs are
given by a vector $\mathbf{c}$.
Then the \textup{frugality ratio} of $\mathcal{M}$ on $\vecc$ is defined as
$\phi_\calM(\vecc) = \frac{p_{_\mathcal{M}}(\mathbf{c})}{\nu(\mathbf{c})}$.
Further, the frugality ratio of $\mathcal{M}$ is defined as
$\phi_\calM = \sup_{\vecc}\phi_\calM(\vecc)$.
\end{definition}

\section{Pruning-Lifting Mechanism}\label{section-mechanism}

In this section, we describe a general scheme for designing truthful mechanisms for set systems,
which we call \mainmechanism. For a given set system $(\calE,\calF)$, the mechanism is composed of the following steps:
\begin{itemize}
\item Pruning. The goal of the pruning process is to drop some elements of
$\calE$ to expose the structure of the competition between the agents;
we denote the set of surviving
agents by $\calE^*$. We require the process to satisfy the
following properties:
    \begin{itemize}
    \item Monotonicity: for any given vector of other agents' bids, if an agent $e$ is
          dropped when he bids $b$, he is also dropped if he bids any $b'>b$.
          We set $t_1(e)=\inf\{b'\mid\text{$e$ is dropped when bidding $b'$}\}$.

    \item Bid-independence: for any given vector of other agents' bids,
          let $b$ and $b'$ be two bids of agent $e$ such that $e$ is
          not dropped when he submits either of them. Then for both
          of these bids the set $\calE^*$ of remaining agents is the same. That is, $e$
          cannot control the outcome of the pruning process as long as he survives.
          Monotonicity and bid-independence conditions are important to ensure
          the truthfulness of the mechanism.

    \item Monopoly-freeness: the remaining set system must remain monopoly-free, i.e.,
          $\bigcap_{S\in\calF^*}S=\emptyset$, where $\calF^*=\{S'\in\calF\mid S'\subseteq \calE^*\}$.
          This condition is necessary because in the winner selection stage
          we will choose a winning feasible set from $\calF^*$. Therefore, we have to make
          sure that no winning agent can charge an arbitrarily
          high price due to lack of competition.
    \end{itemize}

\item Lifting. The goal of the lifting process is to assign a weight to each agent
      in $\calE^*$ in order to take into account the size of each feasible set.
      To this end, construct an undirected graph $\mathcal{H}$ by
      (a) introducing a node $v_e$ for each $e\in \calE^*$, and (b) connecting $v_e$ and
      $v_{e'}$ if and only if any feasible set in $\calF^*$ contains either $e$ or $e'$.
      We will refer to $\calH$ as the {\em dependency graph} of $\calE^*$.
      For each connected component $\mathcal{H}_j$ of $\mathcal{H}$, compute the
      largest eigenvalue $\alpha_j$ of its adjacency matrix $A_j$, and let
      $\left(w(v_e)\right)_{v_e\in\calH_j}$ be the eigenvector of $A_j$
      associated with $\alpha_j$. That is, $A_j\mathbf{w}^j=\alpha_j\mathbf{w}^j$, where
      $\mathbf{w}^j=\left((w(v_e))_{v_e\in\calH_j}\right)^T$. Set $\alpha=\max\alpha_j$.

\item Winner selection.
      Define $b'(e)=\frac{b(e)}{w(v_e)}$ for each $e\in \calE^*$, and select a
      feasible set $S\in \calF^*$ with the smallest total bids w.r.t. $\mathbf{b}'$.
      Let $t_2(e)$ be the threshold bid for $e\in \calE^*$ to be selected
      at this stage.

\item Payment. The payment to each winner $e\in S$ is $p(e)=\min\{t_1(e),t_2(e)\}$, where
      $t_1(e)$ and $t_2(e)$ are the two thresholds defined above.
\end{itemize}

Recall that the largest eigenvalue of the adjacency matrix
of a connected graph is positive and its associated eigenvector
has strictly positive coordinates~\cite{godsil}.
Therefore, $w(v_e)>0$ for all $e\in\calE^*$.

We will now define a quantity  $\alpha_{(\calE, \calF)}$
that will be instrumental in characterizing
the frugality ratio of truthful mechanisms on $(\calE, \calF)$.
Let $\calS(\calE, \calF)$ be the collection of all
monopoly-free subsets of $\calE$, i.e.,
set $\calS(\calE, \calF) = \{S\subseteq \calE\mid \bigcap_{T\in \calF, T\subseteq S}\ T = \emptyset\}$.
The elements of $\calS(\calE, \calF)$ are the possible
outcomes of the pruning stage.
For any subset $S\in\calS(\calE, \calF)$, let $\calH_S$
be its dependency graph and $A_S$ be the adjacency matrix of $\calH_S$.
Let $\alpha_{_S}$ be the largest eigenvalue of $A_S$
(or the maximum of the largest eigenvalues
of the adjacency matrices of the connected components of $\calH_S$,
if $\calH_S$ is not connected).
Set $\alpha_{(\calE, \calF)}=\max_{S\in \calS(\calE, \calF)} \alpha_{_S}$;
we will refer to $\alpha_{(\calE, \calF)}$
as the {\em eigenvalue of the set system $(\calE, \calF)$}.

Note that once $\calE^*\in \calS(\calE, \calF)$ is selected in the pruning step, the
computation of $\alpha$ and the weight vector $(w(v_e))_{e\in\calE^*}$
does not depend on the bid vector.
This property is crucial for showing that our mechanism
is truthful. (Due to space limits, most proofs in the paper are
relegated to the Appendix.)

\begin{theorem}\label{theorem-truthful}
\mainmechanism\ is truthful for any set system $(\calE, \calF)$.
\end{theorem}

In the rest of this section, we will show that the mechanisms
for $r$-out-of-$k$ systems and single path systems proposed in~\cite{anna}
can be viewed as instantiations of our \mainmechanism.
By Theorems~\ref{theorem-monotone} and~\ref{theorem-truthful}, we can ignore the payment
rule in the following discussion.

\subsection{$r$-out-of-$k$ Systems Revisited}\label{section-r-out-of-k}

In an $r$-out-of-$k$ system, the set of agents $\calE$ is a union of $k$
disjoint subsets $S_1,\ldots,S_k$ and the feasible sets
are unions of exactly $r$ of those subsets.
Given a bid vector $\vecb$, renumber the subsets $S_1, \dots, S_k$
in order of non-decreasing bids, i.e.,
$b(S_1)\le b(S_2)\le \cdots \le b(S_k)$.

The mechanism proposed in~\cite{anna} deletes all but the first $r+1$ subsets,
and then solves a system of equations given by
\begin{equation*}
\text{\large{($\lozenge$)}} \quad
\beta = \frac{1}{rx_i}\cdot\sum_{j\neq i}x_j\cdot |S_j|\quad
\text{for $i=1, \dots, r+1$}.
\end{equation*}
It then scales the bid of each set $S_i$ by setting
$b'(S_i)=\frac{b(S_i)}{x_i}$, discards
the set with the highest scaled bid w.r.t. $\mathbf{b}'$, and outputs the remaining sets.

Now, clearly, the first step of this mechanism can be interpreted
as a pruning stage. Further, for $r$-out-of-$k$ systems
the graph $\calH$ constructed in the lifting
stage of our mechanism is a complete $(r+1)$-partite graph. It is not hard
to verify that for any positive solution $(x_1, \dots, x_{r+1}, \beta)$
of the system ($\lozenge$), $\beta\cdot r$
gives the largest eigenvalue of the adjacency matrix of $\calH$
and $(x_1, \dots, x_1, \dots, x_{r+1},\dots, x_{r+1})$ is the corresponding eigenvector.
Thus, the mechanism of~\cite{anna} implements \mainmechanism\
for $r$-out-of-$k$ systems.

In~\cite{anna} it is shown that the frugality ratio of this
mechanism is $\beta$, and that the frugality ratio of any truthful
mechanism for $r$-out-of-$k$ systems is at least $\frac{\beta}{2}$.
As $r$-out-of-$k$ systems can be viewed as a special case of
$r$-path systems, Theorem~\ref{thm:frug-paths}
allows us to improve this lower bound to $\frac{\beta r}{r}=\beta$.

\subsection{Single Path Mechanisms Revisited}\label{section-single-path}

In a single path system, agents are edges of a given directed graph $G=(V, E)$
with two specified vertices $s$ and $t$,
i.e. $\calE=E$ and $\calF$ consists of all sets of edges that contain
a path from $s$ to $t$.

Given a bid vector $\vecb$, the $^{\sqrt{\ }}$-mechanism~\cite{anna}
first selects two edge-disjoint $s$-$t$ paths $P$ and $P'$ that minimize $b(P)+b(P')$.
Assume that $P$ and $P'$ intersect at $s=v_1,v_2,\ldots,v_{\ell+1}=t$, in the order in which they appear in
$P$ and $P'$. Let
$P_i$ and $P'_i$ be the subpaths of $P$ and $P'$ from $v_i$ to $v_{i+1}$, respectively.
The $^{\sqrt{\ }}$-mechanism sets $b'(e)=b(e)\sqrt{|P_i|}$ for $e\in P_i$,
$b'(e)=b(e)\sqrt{|P'_i|}$ for $e\in P'_i$, and chooses a cheapest path
in $P\cup P'$ w.r.t. $\vecb'$.

As in the previous case, the selection of $P$ and $P'$ can be viewed as the pruning
process. The corresponding graph $\calH$ consists of $\ell$ connected components, where
the $i$-th component $\calH_i$ is a complete bipartite graph with parts of size $|P_i|$
and $|P'_i|$. Its largest eigenvalue is given by $\alpha_i=\sqrt{|P_i||P'_i|}$, and
the coordinates of the corresponding eigenvector are given by $w(v_e)=1/\sqrt{|P_i|}$
for $e\in P_i$ and $w(v_e)=1/\sqrt{|P'_i|}$ for $e\in P'_i$. Thus,
the $^{\sqrt{\ }}$-mechanism can be viewed as a special case of the \mainmechanism.
It is shown that the frugality ratio of the $^{\sqrt{\ }}$-mechanism
is within a factor of $2\sqrt{2}$ from optimal;
Theorem~\ref{thm:frug-paths} below shows that this bound can be improved by a factor of $\sqrt{2}$
(this has also been shown by Yan~\cite{yan} via a proof that is considerably more
complicated than ours).

\section{Vertex Cover Systems}\label{sec:vc}
In the vertex cover problem, we are given a graph $G=(V, E)$ whose vertices
are owned by selfish agents. Our goal is to purchase a vertex cover of $G$.
That is, we have $\calE=V$, and $\calF$ is the collection of all vertex covers of $G$.
Let $A$ denote the adjacency matrix of $G$, and let
$\Delta$, $\alpha=\alpha_{(\calE,\calF)}$ and $\vecw=(w(v))_{v\in V}$
denote, respectively, the maximum degree of $G$, the largest eigenvalue of $A$ and the
corresponding eigenvector.

We will use the pruning-lifting scheme to construct a mechanism whose frugality ratio
is $\alpha$; this improves the bound of $2\Delta$ given in~\cite{edith3} by at least a factor of 2
for all graphs, and by as much as a factor of $\Theta(\sqrt{n})$ for some graphs
(e.g., the star).

Observe first that the vertex cover system plays a special
role in the analysis of the performance of the
pruning-lifting scheme. Indeed, on one hand, it is straightforward
to apply the \mainmechanism\ to this system: since removing any agent
will make each of its neighbors a monopolist,
the pruning stage of our scheme is redundant, i.e., $\calH=G$.
That is, there is a unique implementation of \mainmechanism\
for vertex cover systems:
we set $b'(v)=\frac{b(v)}{w(v)}$ for all $v\in V$,
pick any $S\in\arg\min\{b'(T)\mid T\text{ is a vertex cover for $G$}\}$
to be the winning set, and
pay each agent $v\in S$ his threshold bid $t(v)$.
On the other hand, for general set systems, any feasible set in the pruned system
corresponds to a vertex cover of $\calH$: indeed, by construction of the graph $\calH$,
any feasible set must contain at least one endpoint of any edge of $\calH$.
In general, the converse is not true: a vertex cover of $\calH$
is not necessarily a feasible set. However, for $k$-path systems
it is possible to show that any cover of $\calH$ corresponds to a $k$-flow.
We can use this fact to provide an alternative proof of
Theorems~\ref{thm:alpha-nu} and~\ref{thm:alpha-mu}.

We will now bound the frugality of \mainmechanism\
for vertex cover systems.

\begin{theorem}\label{vc-upper}
The frugality ratio of \mainmechanism\  for vertex cover systems
on a graph $G$ is at most $\alpha=\alpha_{(\calE,\calF)}$.
\end{theorem}
\begin{proof}
Since our mechanism is truthful, we have $b(v)=c(v)$ for all $v\in V$.
By optimality of $S$
we have $b'(v)\le \sum_{uv\in E, u\not\in S}b'(u)$, and therefore
$t(v)\le w(v)\sum_{uv\in E, u\not\in S}\frac{c(u)}{w(u)}$.
Thus, we can bound
the total payment of our mechanism given a bid vector $\vecb$ as
\begin{eqnarray*}
\sum_{v\in S}t(v) \le \sum_{v\in S}w(v)\sum_{uv\in E, u\not\in
S}\frac{c(u)}{w(u)} = \sum_{u\notin S}\frac{c(u)}{w(u)}\sum_{uv\in
E}w(v) = \sum_{u\notin S}\frac{c(u)}{w(u)}\alpha w(u) =
\alpha\sum_{u\notin S}c(u).
\end{eqnarray*}
Lemma 8 in~\cite{edith3} shows that for any cost vector $\vecc$
we have $\nu(\vecc)\ge \sum_{u\notin S}c(u)$. Therefore, the frugality
ratio of \mainmechanism\  for vertex cover on $G$ is at most $\alpha$.
\end{proof}
In Section~\ref{vc-lower-bound} we show that our mechanism
is optimal for a large class of graphs.

\subsection{Computational Issues}
To implement
\mainmechanism\  for vertex cover, we need to select the vertex cover
that minimizes the scaled costs given by $(b'(v))_{v\in V}$, i.e.,
to solve an NP-hard problem. However, the argument in
Theorem~\ref{vc-upper} applies to any truthful mechanism
that selects a {\em locally optimal} solution, i.e., a vertex cover $S$
that satisfies $b'(v)\le\sum_{uv\in E, u\not\in S}b'(u)$ for all
$v\in S$. Paper~\cite{edith3} argues that any monotone winner selection
algorithm for vertex cover can be transformed into a locally optimal one,
and shows that a variant of the classic 2-approximation
algorithm for this problem~\cite{baryehuda} is monotone.
This leads to the following corollary.

\begin{coro}
There exists a truthful polynomial-time vertex cover auction
that given a graph $G$ outputs a solution whose cost is within a factor of $2$
from optimal and whose frugality ratio is at most $\alpha$.
\end{coro}

\section{Multiple Paths Systems}\label{section-multi-path}

In this section, we study in detail $k$-paths systems for a given integer $k\ge 1$.
In these systems, the set of agents $\calE$ is the set of edges
of a directed graph $G=(V,E)$ with two specified vertices $s,t\in V$.
The feasible sets are sets of edges that contain $k$ edge-disjoint $s$-$t$ paths.
Clearly, these set systems generalize both $r$-out-of-$k$ systems and single path systems.

Our mechanism for $k$-paths systems for a given directed graph $G$,
which we call \pathmechanism, is a natural generalization of the $^{\sqrt{\ }}$-mechanism~\cite{anna}:
In the pruning stage of our mechanism, given a bid vector $\vecb$, we pick $k+1$
edge-disjoint $s$-$t$ paths $P_1,\ldots,P_{k+1}$ so as
to minimize their total bid w.r.t. the bid vector $\mathbf{b}$.
Clearly, this procedure is monotone and bid-independent.
Let $G^*(\vecb)$ denote the subgraph composed of these $k+1$ paths.
The remaining steps of the mechanism (lifting, winner selection,
payment determination) are the same as in the general case (Section~\ref{section-mechanism}).
Since the \pathmechanism\  is an implementation of the \mainmechanism,
Theorem~\ref{theorem-truthful} implies that it is truthful.

Let $\calG^{k+1}$ denote the set of all subgraphs of $G$
that can be represented as a union of $k+1$ edge-disjoint $s$-$t$ paths in $G$.
For any $G^*\in\calG^{k+1}$, let $\calH(G^*)$ denote the dependency graph of $G^*$,
and let $\alpha(G^*)$ denote the maximum of the largest eigenvalues of the connected
components of $\calH(G^*)$. Set $\alpha_{k+1}=\max\{\alpha(G^*) \mid G^*\in\calG^{k+1}\}$.
We can bound the frugality ratio of our mechanism as follows.

\begin{theorem}\label{thm:alpha-nu}
The frugality ratio of \pathmechanism\ is at most $\alpha_{k+1}\frac{k+1}{k}$.
\end{theorem}

To prove Theorem~\ref{thm:alpha-nu}, we need the following definition.
\begin{definition}[Minimum Longest Path $\delta_{k+1}(G, \vecc)$]
For any $k+1$ edge-disjoint $s$-$t$ paths $P_1,\ldots,P_{k+1}$
in a directed graph $G$,
let $\delta_{k+1}(P_1,\ldots,P_{k+1}, \vecc)$ denote the length of the longest
$s$-$t$ path w.r.t.
cost vector $\mathbf{c}$ in the subgraph $G'$
composed of $P_1,\ldots,P_{k+1}$ (if $G'$ contains a
positive length cycle, set
$\delta_{k+1}(P_1,\ldots,P_{k+1})=+\infty$).
Define
\[
\delta_{k+1}(G, \vecc) =
\min\{\delta_{k+1}(P_1,\ldots,P_{k+1}, \vecc) \mid P_1,\ldots,P_{k+1}
\ \text{\em are $k+1$ edge-disjoint $s$-$t$ paths}\}.
\]
\end{definition}

It turns out that we can give a lower bound on $\nu(\vecc)$
in terms of $\delta_{k+1}(G, \vecc)$.

\begin{lemma}\label{lemma-nash-lower-bound}
For any $k$-paths system on a given graph $G$ with costs $\vecc$,
we have $\nu(\mathbf{c})\ge k\cdot \delta_{k+1}(G, \vecc)$.
\end{lemma}

Let $L(G, \vecc)$ be the length of the longest path in $G^*(\vecc)$,
where $G^*(\vecc)$ is the output of our pruning process on the bid vector $\vecc$.
Our second lemma gives an upper bound on the payment of our mechanism
in terms of $L(G, \vecc)$

\begin{lemma}\label{lem:k-path-upper}
For any $k$-paths system on a given graph $G$ with costs $\vecc$,
the total payment of \pathmechanism\ on $\vecc$ is at most
$\alpha(G^*(\vecc))L(G, \vecc)$.
\end{lemma}

Theorem~\ref{thm:alpha-nu} now follows from
Lemmas~\ref{lemma-nash-lower-bound} and~\ref{lem:k-path-upper}
and the observation that $G^*(\vecc)\in \calG^{k+1}$
and $L(G, \vecc)\le(k+1)\delta_{k+1}(G, \vecc)$ (see Appendix~\ref{appendix-k-path} for details).

In Section~\ref{path-lower-bound}, we show a lower bound of $\frac{\alpha_{k+1}}{k}$
on the frugality ratio of essentially any truthful $k$-path mechanism;
thus the frugality ratio of our mechanism is within a factor of $k+1$
from optimal. This gap leaves the question of whether
there is a better truthful mechanism for $k$-paths systems.
One might hope that a different pruning approach could lead to a smaller
frugality ratio. In particular, the proof of Theorem~\ref{thm:alpha-nu}
suggests that we could get a stronger result by pruning the graph
so as to minimize the length of the longest path $\delta_{k+1}(G^*, \vecc)$
in the surviving graph $G^*$.
While the argument above shows that --- under truthful bidding --- such mechanism
would have an optimal frugality ratio, unfortunately, it turns out
that this pruning process is not monotone~\cite{kempe-bug}.

We can, however, show that our mechanism is {\em optimal} with respect
to a weaker benchmark, namely, one that
corresponds to a buyer-pessimal rather than buyer-optimal
Nash equilibrium. This benchmark was introduced in~\cite{edith3},
and has been recently used by Kempe et al.~\cite{kempe}.
As argued
in~\cite{edith3} and~\cite{kempe}, unlike $\nu$, this benchmark enjoys natural
monotonicity properties and is easier to work with.

\begin{definition}[Benchmark $\mu(\mathbf{c})$~\cite{edith3}]\label{def-new-benchmark}
Given a set system $(\calE, \calF)$, and a feasible set $S\in \calF$ of
minimum total cost w.r.t. $\vecc$,
let $\mu(\mathbf{c})$ be the value of an
optimal solution to the following optimization problem:
\begin{eqnarray*}
\max \ \ & & \sum_{e\in S}b(e) \\
      \text{\em s.t.} \ \ & & (1) \ b(e) \ge c(e) \text{\em \ for all } e\in \calE \\
           & & (2) \ \sum_{e\in S\setminus T} b(e) \le \sum_{e\in T\setminus S} c(e)
                     \text{\em \ for all } \ T\in \calF \\
           & & (3) \ \text{\em For every $e\in S$ there is a $T\in \calF$ s.t. $e\notin T$ and }
                   \sum_{e'\in S\setminus T} b(e') = \sum_{e'\in T\setminus S} c(e')
\end{eqnarray*}
We will refer to the quantity $\sup_{\vecc}\frac{p_{_\mathcal{M}}(\vecc)}{\mu(\vecc)}$
as the {\em $\mu$-frugality ratio} of a mechanism $\calM$, where $p_{_\mathcal{M}}(\mathbf{c})$
is the total payment of mechanism $\mathcal{M}$ on a bid vector $\vecc$.
\end{definition}

The programs for $\nu(\vecc)$ and $\mu(\vecc)$ differ in their objective function only:
while $\nu(\vecc)$ minimizes the total payment, $\mu(\vecc)$ maximizes it.
In particular, this means
that in the program for $\mu(\vecc)$ we can omit constraint~(3), i.e., $\mu(\vecc)$ can be obtained
as a solution to a {\em linear} program.
Kempe et al.~\cite{kempe} show that the $\mu$-frugality ratio
of \pathmechanism\ is within a factor of $2(k+1)$ from optimal.
Our next result, combined with the observation
that our lower bound on the performance of ``all" truthful mechanisms
also holds for the $\mu$-frugality ratio,
shows that \pathmechanism\ is, in fact,
optimal with respect to $\mu$. The proof proceeds by constructing a bid vector $\vecy$
that satisfies constraints~(1) and~(2) in the definition of $\mu(\vecc)$ and pays at least
$kL(G, \vecc)$, and uses the observation that for any fixed network the cost
of a flow of size $x$ is a convex piecewise-linear function of $x$. Details of the proof are given in Appendix~\ref{appendix-mu-frugal}.

\begin{theorem}\label{thm:alpha-mu}
The $\mu$-frugality ratio of \pathmechanism\ is at most $\frac{\alpha_{k+1}}{k}$.
\end{theorem}

\section{Lower Bounds}\label{section-lower-bound}
We say that a mechanism
$\calM$ for a set system $(\calE, \calF)$ is {\em measurable}
if the payment $p(e)$ of any agent $e\in\calE$ is a Lebesque measurable
function of all agents' bids.
We will now use Young's inequality to
give a lower bound on total payments of any measurable
truthful mechanism with bounded frugality ratio.

\begin{theorem}[Young's inequality]
Let $f_1:[0, a]\to\mathbb{R}^+\cup\{0\}$,
$f_2:[0, b]\to\mathbb{R}^+\cup\{0\}$ be two Lebesgue measurable functions
that are bounded on their domain.
Assume that whenever $y>f_1(x)$ for some
$0< x\le a$, $0< y\le b$, we have $x\leq f_2(y)$. Then
$$
\int_0^a f_1(x)dx+\int_{0}^b f_2(y) dy\geq ab .
$$
\end{theorem}

Fix a set system $(\calE, \calF)$ with $|\calE|=n$ and let
$S_{(\calE, \calF)}\in \calS(\calE, \calF)$ be a subset with
$\alpha_{_S}=\alpha_{(\calE, \calF)}$.
For any $e\in S_{(\calE, \calF)}$, let $\vecc_{e, x}$ denote a bid vector
where $e$ bids $x$, all agents in $S_{(\calE, \calF)}\setminus\{e\}$ bid $0$,
and all agents in $\calE\setminus S_{(\calE, \calF)}$
bid $n+1$.

\begin{lemma}\label{lem:lower-bound}
For any set system $(\calE, \calF)$
and any measurable truthful mechanism $\calM$
with bounded frugality ratio, there exists
an agent $e\in S_{(\calE, \calF)}$
and a real value $0<x\le 1$ such that the total payment of $\calM$ on bid vector
$\vecc_{e, x}$ is at least $\alpha_{(\calE, \calF)} x$.
\end{lemma}
\begin{proof}
Set $S=S_{(\calE, \calF)}$, $\calH=\calH_S$, $A=A_S$,
$\alpha = \alpha_{_S} = \alpha_{(\calE, \calF)}$.
We will assume from now on that $\calH=(S, E(\calH))$ is connected;
if this is not the case, our argument can be applied without
change to the connected component of $\calH$ that corresponds to $\alpha$.
Let $\vecw = (w_v)_{v\in S}$ be the eigenvector
of $A$ that is associated with $\alpha$.
By normalization, we can assume that $\max_{v\in S}w_v=1$.

The proof is by contradiction: assume that there is
a truthful mechanism $\calM$ that
pays less than $\alpha x$ on any bid vector of the form
$\vecc_{e, x}$ for all $e\in S$ and all $0<x\le 1$.
Recall that for any such bid vector the cost of each agent
in $\calE\setminus S$ is $n+1$.
Since $\alpha\le n$ and $x\le 1$, this implies that
$\calM$ never picks any agents from $\calE\setminus S$
on any $\vecc_{e, x}$, i.e., effectively $\calM$ operates on $S$.
For any edge $vu$ of $\mathcal{H}$ and any $x>0$,
let $p_{uv}(x)$ denote the payment to $v$ on the bid vector $\vecc_{u, x}$.
Observe that measurability of $\calM$ implies that
$p_{uv}(x)$ is measurable
(since it is a restriction of a measurable function).
In this notation, our assumption can be restated as
\begin{equation}
 \sum_{uv\in E(\calH)}p_{uv}(x)< \alpha x
\end{equation}
for all $u\in S$ and any $0<x\le 1$.

It is easy to see that given a bid vector $\vecc_{u, z}$
with $z\le 1$, $\calM$ never selects $u$ as a winner.
Indeed, suppose that $u$ wins given $\vecc_{u, z}$.
Then by the truthfulness of $\mathcal{M}$, if we reduce $u$'s true cost
from $z$ to $0$, $u$ still wins and receives a payment of at least $z$.
Since the set system restricted to $S$ is monopoly-free, the resulting
cost vector $\vecc'$ satisfies conditions (1)--(3) in the definition of $\nu$,
and hence $\nu(\vecc')=0$.
Thus the frugality ratio of $\mathcal{M}$ is $+\infty$, a contradiction.
By the construction of $\mathcal{H}$,
this means that any $v\in S$ with $uv\in E(\calH)$ wins
given $\vecc_{u, z}$.

Now, fix some $x, y$ such that $0<x, y\le 1$ and
$y>p_{vu}(x)$, and
consider a situation where $v$ bids $x$, $u$ bids $y$, all agents in
$S\setminus\{u, v\}$ bid $0$, and all agents in $\calE\setminus S$ bid $n+1$.
Clearly, in this situation agent $u$ loses and thus $v$ wins
with a payment of at least $x$. By the truthfulness of $\calM$, the same holds
if $v$ lowers his bid to $0$.
Thus, for any $0<x, y\le 1$, $y>p_{vu}(x)$ implies
$p_{uv}(y)\geq x$.

By our assumption, we have  $p_{uv}(x)\le\alpha x$, $p_{vu}(x)\le\alpha x$
for $x\in[0, 1]$. Hence, for any $uv\in E(\calH)$ the functions $p_{uv}(x)$
and $p_{vu}(x)$ satisfy all conditions of
Young's inequality on $[0, 1]$.

Let $A=(a_{uv})_{u, v\in S}$, and consider the
scalar product $\langle \vecw,A\vecw\rangle=\langle \vecw,\alpha \vecw\rangle=
\alpha\langle \vecw,\vecw\rangle$.
We have $\langle \vecw,A\vecw\rangle=\sum_{uv\in E(\mathcal{H})}w_uw_v$.
As we normalized $\vecw$ so that $w_u, w_v\le 1$,
by Young's inequality, we can bound $w_uw_v$ by $\int_{0}^{w_u} p_{uv}(x) dx+\int_{0}^{w_v} p_{vu}(x) dx$.
Therefore,
\begin{eqnarray*}
\alpha\langle \vecw,\vecw\rangle &=&\sum_{u,v\in S} a_{uv} w_uw_v
\leq \sum_{u,v\in S} \left(\int_{0}^{w_u} a_{uv}p_{uv}(x) dx+\int_{0}^{w_v}
a_{uv}  p_{vu}(y)dy \right) \\
&=& 2\sum_{u\in S} \int_0^{w_u} \sum_{v\in S} a_{uv} p_{uv}(x)dx
< 2\alpha\sum_{u\in S} \int_{0}^{w_u} xdx
= \alpha \sum_{u\in S} w_u^2
= \alpha\langle \vecw,\vecw\rangle,
\end{eqnarray*}
where the last inequality follows from (1).
This is a contradiction, so the proof is complete.
\end{proof}
%

\subsection{Vertex Cover Systems}\label{vc-lower-bound}

For vertex cover systems, deleting any of the agents would result in
a monopoly. Therefore, Lemma~\ref{lem:lower-bound} simply says
that for any measurable truthful mechanism $\calM$ on a graph $G=(V, E)$,
there exists a $v\in V$ such that the total payment
on bid vector $x\cdot\vecc_v$ is at least $\alpha x$,
where  $\alpha$ is the largest eigenvalue of the adjacency matrix of $G$ and
$\vecc_v$ is the cost vector given by $c_v(u)=1$ if $u=v$,
and $c_v(u)=0$ if $u\in V\setminus \{v\}$.

Given a graph $G=(V, E)$ and a vertex $v\in V$, let $L_v$ denote
the set of all maximal cliques in $G$ that contain $v$.
Let $\rho_v$ denote the size of the smallest clique in $L_v$.

\begin{lemma}\label{lem:vc-lowerbound}
We have $\nu(x\cdot\vecc_v)\le x(\rho_v-1)$ for any $x>0$.
\end{lemma}

Combining Lemma~\ref{lem:vc-lowerbound} with Lemma~\ref{lem:lower-bound} yields the following result.

\begin{theorem}\label{thm:frug-vc}
For any graph $G$, the frugality ratio of any measurable truthful vertex cover
auction on $G$ is at least $\frac{\alpha}{\rho}$, where $\alpha$
is the largest eigenvalue of the adjacency matrix of $G$, and
$\rho=\max_{v\in V}\rho_v$.
\end{theorem}

The bound given in Lemma~\ref{lem:vc-lowerbound} is not necessarily optimal;
we can construct a family of graphs where for some vertex $v$
the quantity $\rho_v$ is linear in the size of the graph, while
$\nu(\vecc_v)=O(1)$. Nevertheless, Theorem~\ref{thm:frug-vc}
shows that the mechanism described in Section~\ref{sec:vc} has optimal
frugality ratio for, e.g., all triangle-free graphs and, more generally,
all graphs $G$ such that the for each vertex $v\in G$, the subgraph induced
on the neighbors of $v$ contains an isolated vertex.

\subsection{Multiple Path Systems}\label{path-lower-bound}

Let $(\calE, \calF)$ be a $k$-path system on a graph $G=(V, E)$.
Consider a set $S\in\calS(\calE, \calF)$ with
$\alpha_{_S}=\alpha_{(\calE, \calF)}$. It is not hard to see
that $S$ is a union of $k+1$ edge-disjoint paths;
this follows, e.g., from the proof
of Theorem~\ref{them-nash-k-path}.
Hence, we have $\alpha_{(\calE, \calF)}=\alpha_{k+1}$.

As before, for any $e\in S$, let $\vecc_{e, x}$ denote the cost vector
with $\vecc_{e, x}(e)=x$, $\vecc_{e, x}(u)=0$ for all $u\in S\setminus\{e\}$,
$\vecc_{e, x}(w)=n+1$ for all $w\in E\setminus S$. It is easy to see
that we have $\mu(\vecc_{e, x})= \nu(\vecc_{e, x}) =kx$ for any $e\in S$, $x>0$.
Combining this observation with Lemma~\ref{lem:lower-bound},
we obtain the following result.
\begin{theorem}\label{thm:frug-paths}
For any graph $G=(V, E)$,
both the frugality ratio
and the $\mu$-frugality ratio of any measurable truthful $k$-path auction
on $G$ are at least $\frac{\alpha_{k+1}}{k}$.
\end{theorem}
In Section~\ref{section-multi-path}, we show that the
frugality ratio and the $\mu$-frugality ratio of \pathmechanism\
are bounded by, respectively, $\alpha_{k+1}\frac{k+1}{k}$ and $\frac{\alpha_{k+1}}{k}$.
Together with Theorem~\ref{thm:frug-paths} this implies that \pathmechanism\
has optimal $\mu$-frugality ratio, and its frugality ratio is within a factor of
$(k+1)$ from optimal.

\section{Conclusions and Open Problems}

In this paper, we propose a uniform scheme for designing frugal truthful mechanisms.
We show that several existing mechanisms can be viewed as instantiations of our scheme,
and describe its applications to $k$-path systems and vertex cover systems.
We demonstrate that our scheme produces mechanisms with good frugality ratios
for $k$-path systems and a large subclass of vertex cover systems;
for $k$-path systems, we show that our mechanism has the optimal frugality ratio.
Moreover, all mechanisms described in this paper are polynomial-time computable.
We believe that our scheme can be applied to many other set systems, resulting in
mechanisms with near-optimal frugality ratios.

It would be interesting to understand the limits of applicability of our scheme.
Indeed, for some set systems the minimal monopoly-free subsystem does not
necessarily exhibit a lot of connections between agents, i.e., the corresponding
dependency graph is rather sparse. It seems that for such cases our scheme does
not produce mechanisms with good frugality ratio. Formalizing this intuition
and developing alternative approaches for designing frugal mechanisms in such settings
is an interesting research direction.

\section{Acknowledgements}

We thank David Kempe, Christos Papadimitriou, and Yaron Singer for helpful discussions.

\newpage

\input{arxiv2-app}

\end{document}

%% file: arxiv2-app.tex
\appendix

\section{Proof of Theorem~\ref{theorem-truthful}}

\begin{proof}
For any agent $e\in \calE$ and given bids of other agents, we will analyze the utility of $e$ in terms of his
bid. There are the following two cases.
\begin{description}
\item[Case 1.] Agent $e$ is not dropped out in the pruning process when bidding $b(e)=c(e)$, i.e. $e\in
\calE^*$.
By the definition of $t_1(e)$, we know that $t_1(e)\ge c(e)$. Consider the situation where $e$ bids another
value $b'(e)\neq b(e)$. If $b'(e)>t_1(e)$, then $e\notin \calE^*$ and his utility is 0. If $b'(e)\le t_1(e)$,
by
the bid-independence property, we know that the subset $\calE^*$ remains the same. Given this fact, the
structure
of the graph $\mathcal{H}$ does not change, which implies that the eigenvectors and eigenvalues of its
adjacency matrix do not change either. Hence,
the threshold value $t_2(e)$ will not change, which implies that the payment to agent $e$,
$p(e)=\min\{t_1(e),t_2(e)\}$, will not change.

\item[Case 2.] Agent $e$ is dropped out in the pruning process when bidding $b(e)=c(e)$, i.e. $e\notin
\calE^*$. Consider the situation where $e$ bids another value $b'(e)\neq b(e)$ and is not dropped out. By
monotonicity and bid-independence, we know that $b'(e)\le t_1(e) \le b(e)=c(e)$. Hence, even though
$e$ could be a winner by bidding $b'(e)$, his payment is at most $t_1(e)\le c(e)$, which implies that he
cannot obtain a positive utility.
\end{description}
From the above two cases, we know that the utility of each agent is maximized by bidding his true cost, and
hence the mechanism is truthful.
\end{proof}

\section{Analysis of \pathmechanism}\label{appendix-k-path}

\subsection{Lower Bound on $\nu$}
Our analysis of \pathmechanism\ relies on the characterization
of Nash flows presented in~\cite{ningnick}.

\begin{theorem}[\cite{ningnick}]\label{them-nash-k-path}
Let $G = (V,E)$ be a directed graph with weight $w(e)$ on each edge $e\in E$. Given two
specific vertices
$s,t\in V$, assume that there are $k$ edge-disjoint paths from $s$ to $t$. Let
$P_1,P_2,\cdots, P_k$ be
such $k$ edge-disjoint $s$-$t$ paths so that its total weight $L\triangleq
\sum_{i=1}^{k}w(P_i)$ is
minimized, where $w(P_i)=\sum_{e\in P_i}w(e)$. Further, it is known that for every edge
$e\in E$, the graph
$G-\{e\}$ has $k$ edge-disjoint $s$-$t$ paths with the same total weight $L$. Then
essentially there exist
$k+1$ edge-disjoint $s$-$t$ paths in $G$ such that each of them is a shortest path from
$s$ to $t$.
\end{theorem}

We will first show how Theorem~\ref{them-nash-k-path} implies Lemma~\ref{lemma-nash-lower-bound}.
For completeness, we will then present a proof of Theorem~\ref{them-nash-k-path};
this proof also appears in~\cite{ningnick}.

\begin{proof}[Proof of Lemma~\ref{lemma-nash-lower-bound}]
Fix a cost vector $\vecc$.
Let $E'$ be the winning set with respect to $\vecc$,
and consider a bid vector $\vecb$ that satisfies
conditions (1)--(3) in the definition of $\nu(\vecc)$. Let $p(\mathbf{b})$
denote the total payment under $\vecb$.
The set $E'$ contains $k$ edge-disjoint $s$-$t$ paths. By condition (2),
no agent in $E'$ can obtain more revenue by increasing his bid.
That is, for any $e\in E'$, there are $k$ edge-disjoint $s$-$t$ paths in
$G\setminus\{e\}$ with the same total bid as $E'$.
Applying Theorem~\ref{them-nash-k-path} with $w(e)=b(e)$,
we obtain that there are $k+1$ edge-disjoint shortest $s$-$t$
paths with length $\frac{p(\mathbf{b})}{k}$ each w.r.t $\mathbf{b}$.
Consider the subgraph $G'$ composed by these $k+1$ edge-disjoint paths.
We know that $\delta_{k+1}(G', \vecc) \le \frac{p(\mathbf{b})}{k}$
as $b(e)\ge c(e)$ for any edge $e$,
i.e. the longest $s$-$t$ path in $G'$ w.r.t to $\mathbf{c}$
is at most $\frac{p(\mathbf{b})}{k}$.
Hence,
\[
p(\mathbf{b}) \ge k\cdot \delta_{k+1}(G', \vecc) \ge k\cdot \delta_{k+1}(G, \vecc).
\]
As this holds for any vector $\vecb$ that satisfies conditions (1)--(3),
it follows that $\nu(\vecc)\ge \delta_{k+1}(G, \vecc)$.
\end{proof}

In the rest of this section, we prove Theorem~\ref{them-nash-k-path}.
\begin{proof}[Proof of Theorem~\ref{them-nash-k-path}]
Given the graph $G$ and integer $k$, we construct a flow network $\mathcal{N}_k(G)$ as
follows: We introduce two extra nodes $s_0$ and $t_0$ and two edges $s_0s$ and $tt_0$. The
set of vertices of $\mathcal{N}_k(G)$ is $V\cup \{s_0,t_0\}$ and the set of edges is
$E\cup \{s_0s, tt_0\}$. The capacity $cap(\cdot)$ and cost per bulk capacity $cost(\cdot)$
for each edge in $\mathcal{N}_k(G)$ is defined as follows:
\begin{itemize}
\item $cap(s_0s) = cap(tt_0) = k$ and $cost(s_0s) = cost(tt_0) = 0$.
\item $cap(e) = 1$ and $cost(e) = w(e)$, for $e\in E$.
\end{itemize}

Given the above construction, every path from $s$ to $t$ in $G$ naturally corresponds to a
bulk flow from $s_0$ to $t_0$ in $\mathcal{N}_{k}(G)$. Hence, the set of $k$ edge-disjoint
paths $P_1,P_2, \ldots, P_k$ in $G$ corresponds to a flow $\mathcal{F}_{G}$ of size $k$ in
$\mathcal{N}_{k}(G)$. In addition, the minimality of $L = \sum_{i=1}^{k}w(P_i)$ implies
that $\mathcal{F}_{G}$ achieves the minimum cost (which is $L$) for all
\textit{integer}-valued flows of size $k$, i.e. maximum flow in $\mathcal{N}_{k}(G)$.
Since all capacities of $\mathcal{N}_{k}(G)$ are integers, we can conclude that
$\mathcal{F}_{G}$ has the minimum cost among all \textit{real} maximum flows in
$\mathcal{N}_{k}(G)$. (The proof of this fact can be found in,
e.g.,~\cite{MinCostMaxFlow}.)

For simplicity, we denote the subgraph $G-\{e\}$ by $G-e$. By the fact that for any $e\in
E$, the subgraph $G-e$ has $k$ edge-disjoint $s$-$t$ paths with the same total weight $L$,
we know that in the network $\mathcal{N}_{k}(G-e)$, there still is an integer-valued flow
$\mathcal{F}_{G-e}$ of size $k$ and cost $L$. So $\mathcal{F}_{G- e}$ is also an
integer-valued flow of size $k$ and cost $L$ in $\mathcal{N}_{k}(G)$. Define a real-valued
flow in $\mathcal{N}_{k}(G)$ by $\mathcal{F} = \frac{1}{|E|}\sum_{e\in E}\mathcal{F}_{G -
e}$. We have the following observations:
\begin{enumerate}
\item It is clear that $\mathcal{F}(e) \le cap(e)$ for every arc $e\in
\mathcal{N}_{k}(G)$, where $\mathcal{F}(e)$ is the amount of flow on edge $e$ in
$\mathcal{F}$, as we have taken the arithmetic average of the flows in the network.

\item $\mathcal{F}$ has cost $\frac{1}{|E|}\sum\limits_{e\in E} cost(\mathcal{F}_{G - e})
= \frac{1}{|E|}\cdot |E|\cdot L = L$.

\item Since $\mathcal{F}_{G-e}(s_0s) = k$ for any $e\in E$, we have $\mathcal{F}(s_0s) =
k$. In addition, as each $\mathcal{F}_{G-e}$ is a feasible flow that satisfies all
conservation conditions and $\mathcal{F}$ is defined by the arithmetic average of all
$\mathcal{F}_{G-e}$'s, we know that $\mathcal{F}$ also satisfies all conservation
conditions.
\end{enumerate}
Therefore, $\mathcal{F}$ is a minimum cost maximum flow in $\mathcal{N}_{k}(G)$. In
addition, $\mathcal{F}$ has the following nice property, which plays a fundamental role
for the proof:
\begin{itemize}
\item \label{Has more flow}
For every edge $e\in \mathcal{N}_{k}(G)$ except $s_0s$ and $tt_0$, we have
$\mathcal{F}(e) \leq cap(e) - \frac{1}{|E|}$, as $\mathcal{F}_{G - e}$ does not flow
through $e$, i.e. $\mathcal{F}_{G - e}(e) = 0$, and $\mathcal{F}_{G - e'}(e)$ is either 0
or 1 for any $e'\in E$.
\end{itemize}

Let $E_{+}=\{e\in\mathcal{N}_{k}(G)\mid\mathcal{F}(e)>0\}$. Suppose that there is a path
$P' = (e_1,e_2,\ldots, e_r)$ from $s_0$ to $t_0$ which goes only along arcs in $E_{+}$ and
is not a shortest path w.r.t $cost(\cdot)$ from $s_0$ to $t_0$ in $\mathcal{N}_{k}(G)$.
Let $\epsilon = \min\left\{\mathcal{F}(e_1),\mathcal{F}(e_2),\ldots,\mathcal{F}(e_r),
\frac{1}{|E|}\right\}$. Since $P'\subseteq E_{+}$, we have $\epsilon > 0$.
Let $P$ be a shortest path w.r.t $cost(\cdot)$ from $s_0$ to $t_0$ in
$\mathcal{N}_{k}(G)$. Define a new flow $\mathcal{F'}$ from $\mathcal{F}$ by adding
$\epsilon$ amount of flow on path $P$ and removing $\epsilon$ amount of flow from path
$P'$. We have the following observations about $\mathcal{F'}$:
\begin{enumerate}
\item The size of flow $\mathcal{F'}$ is $k$.

\item $\mathcal{F'}$ satisfies all conservation conditions as it is a linear combination
of three flows from $s_0$ to $t_0$.

\item By the definition of $\epsilon$, the amount of flow of each edge is non-negative in
$\mathcal{F'}$. Further, $\mathcal{F'}$ satisfies the capacity constrains. This follows
from the facts that $\epsilon \le \frac{1}{|E|}$ and the above property established for
$\mathcal{F}$.

\item The cost of $\mathcal{F'}$ is smaller than $L$ because $cost(\mathcal{F'}) =
cost(\mathcal{F}) - \epsilon(cost(P')-cost(P))$, which is smaller than
$L=cost(\mathcal{F})$ as $cost(P)<cost(P')$ by the assumption.
\end{enumerate}
Hence, $\mathcal{F'}$ is a flow of size $k$ in $\mathcal{N}_{k}(G)$ with cost smaller than
$\mathcal{F}$, a contradiction.
Thus, every path from $s_0$ to $t_0$ in $\mathcal{N}_{k}(G)$ along the edges in $E_{+}$ is
a shortest path w.r.t $cost(\cdot)$.

Consider a new network $\mathcal{N'}_{k + 1}(G)$ obtained from $\mathcal{N}_{k + 1}(G)$ by
restricting edges in $E_{+}$. (Note that the only difference between $\mathcal{N}_{k +
1}(G)$ and $\mathcal{N}_{k}(G)$ is the capacity on edges $s_0s$ and $tt_0$ is $k+1$ rather
than $k$.) We claim that in this network there is an integer-valued flow of size $k+1$.
Suppose otherwise, by Ford-Fulkerson Theorem~\cite{FF}, there is a cut $(S_{s_{0}},
T_{t_{0}})$ in $\mathcal{N'}_{k + 1}(G)$ with size less than or equal to $k$. By
definition, in $\mathcal{N'}_{k + 1}(G)$ we have $cap(s_0s) = k + 1$ and $cap(tt_0) =
k+1$, which implies that $s_0,s\in S_{s_{0}}$ and $t_0,t\in T_{t_{0}}$.
By the definition of $E_+$, we know that the total amount of flows of $\mathcal{F}$ on the
cut $(S_{s_{0}},T_{t_{0}})$ is $k$. Since $\mathcal{F}(e)< 1$ for any edge $e$, we can
conclude that there are at least $k+1$ edges from $S_{s_{0}}$ to $T_{t_{0}}$ in $E_{+}$.
This leads to a contradiction, because we have showed that the size of the cut
$(S_{s_{0}}, T_{t_{0}})$ is less than or equal to $k$.

Therefore, we can find an integer-valued flow of size $k+1$ on edges in $E_{+}$ in the
network $\mathcal{N}_{k + 1}(G)$. Such a flow can be thought as the union of $k+1$
edge-disjoint paths from $s_0$ to $t_0$. We know that every such path going along edges in
$E_{+}$ is a shortest path from $s_0$ to $t_0$. This in turn concludes the proof, since we
have found $k+1$ edge-disjoint shortest paths from $s$ to $t$ in $G$.
\end{proof}

\begin{figure}[ht]
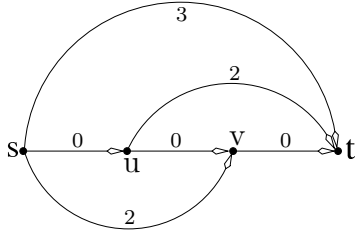

\begin{center}
\include{2paths-arxiv}
\caption{An instance of $1$-path problem on which
\pathmechanism\ does not choose the path that minimizes $\delta_2(G, \vecc)$}
\label{fig:2paths}
\end{center}
\end{figure}

\subsection{Upper Bound on the Payment of \pathmechanism\ }

In this subsection, we will prove Lemma~\ref{lem:k-path-upper}.
Fix a cost vector $\vecc$ and set $G^*=G^*(\vecc)$,
$\calH=\calH(G^*(\vecc))$, $\alpha=\alpha(G^*(\vecc))$.
Observe that since $G^*$ is the cheapest collection of $k$ edge-disjoint
paths in $G^*$, it is necessarily cycle-free.
We say that a vertex $v$ is an {\em articulation point} for $G^*$ if
it lies on any $s$-$t$ path in $G$. Suppose that  $G^*$
has $a+1$ \textit{articulation points}
$s=v_1, v_2, \dots, v_{a+1}=t$ that subdivide $G^*$
into $a$ parts $G^*_1,\ldots,G^*_a$. Denote by $\calH_j$
the induced subgraph of $\calH$ with vertex set $\{v_e\mid e\in G^*_j\}$,
$j=1, \dots, a$. For any graph $\Gamma$, let $V(\Gamma)$ and $E(\Gamma)$ denote
the sets of vertices and edges of $\Gamma$, respectively.
The following intermediate lemma describes the structure
of the graph $\calH$.

\begin{figure}[ht]
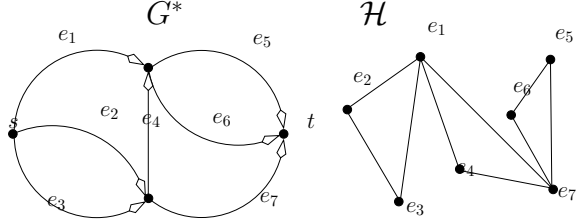

\begin{center}
\include{flow-arxiv}
\caption{An example of the construction of $\mathcal{H}$ from $G^*$ for $k=2$}
\end{center}
\end{figure}

\begin{lemma}\label{cl3}
Graph $\mathcal{H}$ is connected if and only if $G^*$ has no articulation points.
Hence, for each $j$, $\mathcal{H}_j$ is connected.
\end{lemma}

For each vertex $v\in V(\mathcal{H})$, let $e_v$ be the corresponding edge in $G^*$.
Before we prove Lemma~\ref{cl3}, we need the following two observations.

\begin{claim}\label{cl1}
Let $u$ and $v$ be two vertices of $\mathcal{H}$, then $uv\notin E(\mathcal{H})$ if and only if
there is an $s$-$t$ path in $G^*$ going through both $e_u$ and $e_v$.
\end{claim}
\begin{proof}
If there is a path $P\subseteq G^*$ such that $e_u,e_v\in P$, then in $G^*\setminus\{e_u,e_v\}$
there are $k$ edge-disjoint $s$-$t$ paths. Hence there is no edge between $u$ and $v$. Conversely,
if $uv\notin E(\mathcal{H})$, then $G^*\setminus\{e_u,e_v\}$ has $k$ edge-disjoint $s$-$t$
paths. Removing these $k$ paths from $G^*$ leads to an $s$-$t$ path going through $e_u$ and $e_v$.
\end{proof}

\begin{claim}\label{cl2}
Let $P$ be an $s$-$t$ path in $G^*$, corresponding to the vertex sequence $v_1,\ldots,v_\ell$ of
$\mathcal{H}$, and let $v$ be another vertex of $\mathcal{H}$. Then the following holds:
\begin{enumerate}
\item There are integral $i$ and $j$, $1\leq i\leq j\leq \ell$, such that $vv_r\in E(\mathcal{H})$
for $i\leq r\leq j$, and $vv_r\notin E(\mathcal{H})$ for $1\le r < i$ and $j < r \le \ell$.

\item Let $\mathcal{L}(v)=\{v_r \ | \ 1\leq r\leq \ell, vv_r\notin E(\mathcal{H})\}$, then in $G^*$
there is a $s$-$t$ path containing all vertices of $\mathcal{L}(v)$ and $v$.
\end{enumerate}
\end{claim}
\begin{proof}~
\begin{enumerate}
\item Since $G^*\setminus \{e_v\}$ contains only $k$ edge-disjoint $s$-$t$ paths, there is a cut of
size $k+1$ in $G^*$ containing $e_v$ and this cut must contain an edge of $P$. Thus, there exists a
vertex $v_i$ with $1\leq i\leq \ell$ adjacent to $v$ in $\mathcal{H}$. For the first part of the
claim, it remains to show that if $v$ is adjacent to $v_p$ and $v_q$ with $p<q$, then so is
$v_r$
for any $r$ with $p\leq r\leq q$. Suppose otherwise, then by Claim~\ref{cl1} there should be a path
$P'$ going through $e_v$ and $e_{v_r}$ for some $p< r < q$. If $P'$ goes through $e_v$ earlier than
$e_{v_r}$, we can construct a path in $G^*$ going through both $e_v$ and $e_{v_q}$.
Indeed, we can concatenate the prefix of $P'$ that ends with $e_{v_p}$ and the suffix
of $P$ that starts with $e_{v_p}$; this path is simple since $G^*$ is cycle-free.
By a similar argument, if $P'$ goes
through $e_{v_r}$ earlier than $e_v$, then we can construct a path in $G^*$ gong through $e_{v_p}$
and $e_v$. For both cases we arrive at a contradiction with Claim~\ref{cl1}, since $v_p$ and $v_q$
are adjacent to $v$ in $\mathcal{H}$.

\item This part can be derived from the first part and Claim~\ref{cl1} applied to two pairs of
vertices $v_{i-1},v$ and $v_{j+1},v$.
\end{enumerate}
\end{proof}

We are now ready to prove Lemma~\ref{cl3}.

\begin{proof}[Proof of Lemma~\ref{cl3}]
Suppose that $G^*$ contains an articulation point $v\in V(G^*)$. Then we can split $G^*$ into the
union of two subgraphs $G_1$ (from $s$ to $v$) and $G_2$ (from $v$ to $t$). By the definition of
$\mathcal{H}$, it can be seen that for any $e_1\in E(G_1)$ and $e_2\in E(G_2)$,
$v_{e_1}v_{e_2}\notin E(\mathcal{H})$. Hence, if $\mathcal{H}$ is connected, then $G^*$ has no
articulation points.

On the other hand, suppose that $G^*$ has no articulation points. We can endow the set of nodes of
$G^*$ with a partial order by putting $u > v$ if there is a path in $G^*$ from $u$ to $v$. Suppose
that $\mathcal{H}$ is disconnected. Let $C$ be a connected component of $\mathcal{H}$. Let $S(C)$ be
the set of endpoints of the edges in $G^*$ corresponding to those vertices in $C$. We claim that
$S(C)$ has a unique smallest element, denoted by $\rho(C)$. Suppose otherwise that $S(C)$ has at
least two different smallest elements $v_1$ and $v_2$ with $e_1=u_{1}v_{1}\in E(G^*)$ and
$e_2=u_{2}v_{2}\in E(G^*)$. Assume without loss of generality that
$v_1\neq t$. Then there is $v_3\in V(G^*)$ such that $e_3=v_{1}v_3\in E(G^*)$. There could not be a
path going through both $e_2$ and $e_3$, since
in this case either $v_2>v_1$ or $v_1>v_2$. Hence $e_2$ and $e_3$ are adjacent in $\mathcal{H}$,
which implies that $e_3\in C$. But then we have $v_3\in S(C)$ and $v_3<v_1$, a contradiction to the
assumption that $v_1$ is a smallest element.

Since $e_1=uv\in E(G^*)$ and $e_2=wv\in E(G^*)$ are necessarily adjacent in $\mathcal{H}$, two
different connected components $C_1$ and $C_2$ of $\mathcal{H}$ cannot have the same smallest
element, i.e. $\rho(C_1)\neq \rho(C_2)$. Let us take the component $C$ with the smallest element
$\rho(C)\neq t$. It is also clear by the definition of $\rho(C)$ that
$\rho(C)\neq s$. Let $e$ be a vertex of $C$ with the ending point at $\rho(C)$ in $G^*$. We want to
show that $\rho(C)$ is an articulation point for $G^*$. Suppose otherwise that there is a path $P$
in
$G^*$ which does not go through $\rho(C)$. Let $v_e\in\mathcal{H}$ correspond to $e$ and
$v_1,\ldots,v_\ell\in\mathcal{H}$ correspond to $P$. Applying Claim~\ref{cl2} to $P$ and $v_e$, we
get that $v_j$ is adjacent to $v_e$ in
$\mathcal{H}$ and there is a path from endpoint of $v_e$ in $G^*$, i.e. $\rho(C)$, to the endpoint
of $v_j$. Thus we
arrive at a contradiction, since $v_j\in C$ and endpoint of $v_j$ is smaller than $\rho(C)$.
\end{proof}

We will now prove Lemma~\ref{lem:k-path-upper}.

\begin{proof}[Proof of Lemma~\ref{lem:k-path-upper}]
Suppose first that $G^*$ has no articulation points.
Let $P^*$ be a losing path in $G^*$ according to the mechanism.
Then $P^*$ is the most expensive path
in $G^*$ w.r.t. scaled cost vector.
Assume that $P^*$ is represented as a sequence of vertices
$v_1,v_2,\ldots,v_\ell$ in $\mathcal{H}$ and let $v$ be another vertex of $\mathcal{H}$. Let $i(v)$
and $j(v)$ be the two integers, $1\leq i(v)\leq j(v)\leq \ell$, as defined in Claim~\ref{cl2}. Note
that the threshold bid of $v$ is at most $\sum_{r =
i(v)}^{j(v)}\frac{c(v_r)}{w(v_r)}w(v)$, since otherwise by Claim~\ref{cl2}, one can find a more
expensive path (w.r.t. scaled costs) going through $\mathcal{L}(v)$ and $v$ than $P^*$.
Thus we have the following upper bound on the total payment of the mechanism:
\begin{eqnarray*}
\sum\limits_{v\notin P^*}\sum\limits_{r=i(v)}^{j(v)}\frac{c(v_r)}{w(v_r)}w(v) &=&
\sum\limits_{r=1}^{\ell}\sum\limits_{v\in N_{\mathcal{H}}(v_r)}\frac{c(v_r)}{w(v_r)}w(v) \\ &=&
\sum\limits_{r=1}^{\ell}\frac{c(v_r)}{w(v_r)}\sum\limits_{v\in N_{\mathcal{H}}(v_r)}w(v) \\ &=&
\sum\limits_{r=1}^{\ell}\frac{c(v_r)}{w(v_r)}\alpha(G^*) w(v_r) \\ &=&
\alpha(G^*)\sum\limits_{t=1}^{\ell}c(v_t).
\end{eqnarray*}
As $L(G, \vecc)\ge \sum\limits_{t=1}^{\ell}c(v_t)$, this proves the lemma for this case.

For the general case where there are articulation points in $G^*$, by using the same analysis on
each connected component $\mathcal{H}_i$ of $\mathcal{H}$,
we can upper-bound the total payment of our mechanism
by $\sum\limits_{\mathcal{H}_i}\alpha_i\sum\limits_{r=1}^{\ell_i}c(v_{i,r})$,
where $v_{i,1},v_{i,2},\ldots,v_{i,\ell_i}$ are a sequence of vertices in each $\mathcal{H}_i$
corresponding to the most expensive path $P^*$, and $\alpha_i$ is the maximum
eigenvalue of the adjacency matrix of ${\mathcal{H}_i}$.
As $\alpha(G^*)=\max_i \alpha_i$, and the most expensive path in $G^*$ is a concatenation
of the most expensive paths in $G_1^*, \dots, G_a^*$, the lemma is proven.
\end{proof}

\subsection{Proof of Theorem~\ref{thm:alpha-nu}}

We will now show how Lemmas~\ref{lemma-nash-lower-bound} and~\ref{lem:k-path-upper}
imply Theorem~\ref{thm:alpha-nu}.

\begin{proof}[Proof of Theorem~\ref{thm:alpha-nu}]
Fix an arbitrary cost vector $\vecc$.
Suppose that in the pruning stage we pick
a graph $G^*$.
By Lemma~\ref{lem:k-path-upper}, the total payment
of our mechanism is at most $\alpha(G^*)L(G, \vecc)$.
Since $G^*\in\calG^{k+1}$, we have $\alpha(G^*)\le \alpha_{k+1}$.
Consider a collection $P_1, \dots, P_{k+1}$
of $k+1$ edge-disjoint paths in $G$ such that
$\delta_{k+1}(P_1, \dots, P_{k+1}, \vecc) = \delta_{k+1}(G, \vecc)$.
Since $G^*$ is the cheapest collection of $k+1$ edge-disjoint
paths in $G$, we have $\sum_{e\in G^*}c(e)\le\sum_{i=1}^{k+1}\sum_{e\in P_i}c(e)$.
We obtain
$$
L(G, \vecc)\le \sum_{e\in G^*}c(e)\le \sum_{i=1}^{k+1}\sum_{e\in P_i}c(e)
\le(k+1)\delta_{k+1}(P_1, \dots, P_{k+1}, \vecc)=
(k+1)\delta_{k+1}(G, \vecc),
$$
where the last inequality follows from the definition of
$\delta_{k+1}(P_1, \dots, P_{k+1}, \vecc)$.
Thus, the frugality ratio of \pathmechanism\ on $\vecc$
is at most
$$
\frac{\alpha(G^*)L(G, \vecc)}{k\delta_{k+1}(G, \vecc)}\le
\frac{\alpha_{k+1}(k+1)\delta_{k+1}(G, \vecc)}{k\delta_{k+1}(G, \vecc)} =
\frac{\alpha_{k+1}(k+1)}{k}.
$$
\end{proof}


\section{$\mu$-Frugality Analysis: Proof of Theorem~\ref{thm:alpha-mu}}\label{appendix-mu-frugal}
To prove Theorem~\ref{thm:alpha-mu}, we need to show that we can
lower-bound $\mu(\vecc)$ in terms of $L(G, \vecc)$.

Consider an arbitrary network $\calF$ with source $s$, sink $t$,
integer edge capacities and costs per unit flow that are
given by a vector $\vecc$.
Let $M$ be the size of the maximum flow in $\calF$.
For any (real) $x\in[0, M]$, let $C(x)$ be the cost of a cheapest
flow of size $x$ in $\calF$ (i.e. the sum of costs on all edges,
where the cost on an edge $e$ is the amount of flow times $\vecc(e)$). The following lemma
establishes several properties of the function
$C(x)$ that will be used in our proof.

\begin{lemma}\label{lem:c(x)}~
\begin{enumerate}
\item $C(x)$ is a convex function on $[0, M]$.
\item For any integer $i\le M-1$, $C(x)$ is a linear function on the interval $[i,i+1]$.
\end{enumerate}
\end{lemma}
\begin{proof}~
\begin{itemize}
\item \textbf{Convexity.} It suffices to show that for any $0\leq\alpha\leq 1$ we have
      $$
       \alpha C(x)+(1-\alpha)C(y)\geq C(\alpha x+ (1-\alpha) y).
      $$
      Let $f_x$ and $f_y$ be cheapest flows of size $x$ and $y$, respectively.
      Clearly, their respective costs are $C(x)$ and $C(y)$. Then
      the flow $f=\alpha f_x+(1-\alpha)f_y$ is
      a flow of size $(\alpha x+(1-\alpha)y)$ and cost $\alpha C(x)+(1-\alpha) C(y)$
      that satisfies all capacity constraints.
      Thus,
      $$
       \alpha C(x)+(1-\alpha)C(y)= \vecc(f)\geq C(\alpha x+ (1-\alpha)y).
      $$

\item \textbf{Linearity on intervals.}
      We first show that $C(x)$ is linear on the
      interval $[0,1]$. Let us fix $x_0\in[0,1]$ and let $f$ be a cheapest
      flow of size $x_0$.
      We can represent $f$ as a finite sum of positive
      flows along $s$-$t$ paths $p_1, \dots, p_l$, i.e.,
      $$
        f =\sum_{i=1}^{l}\varepsilon_i p_i.
      $$

      We know that $C(1)$ is the cost of cheapest path $p$. Thus we have
      $$
       C(x_0)=\vecc(f)\geq \vecc(p)\sum_{i=1}^{l}\varepsilon_i = C(1)x_0.
      $$
      On the other hand, since $C(x)$ is convex, we have
      $x_0 C(1)+(1-x_0)C(0)\geq C(x_0)$. Hence $C(x_0) = x_0 C(1)$.

      In general, for the interval $[i, i+1]$ we first take a
      cheapest $i$-flow $f_i$ (which we can choose to be integer)
      and consider the residual network $\calF_i=\calF -f_i$.
      We can then apply the argument for the $[0,1]$ case to $\calF_i$.
\end{itemize}
\end{proof}

Our next lemma bounds $\mu(\vecc)$ in terms of the difference
between the cost of the cheapest flow of size $k$ and that of the cheapest
flow of size $k+1$.

\begin{lemma}\label{lem:(k+1)-k}
Let $(\calE, \calF)$ be a $k$-path system given by a directed
graph $G=(V, E)$, source $s$ and sink $t$,
and let $\vecc$ be its cost vector.
Then for the function $C(x)$ defined above
we have $k\cdot (C(k+1)-C(k))\leq \mu(\vecc)$.
\end{lemma}
\begin{proof}
For any cost vector $\vecy\in{\mathbb R}^{|E|}$,
let $C_\vecy(x)$ denote the cost of the cheapest flow
of size $x$ in $G$ with respect to the cost vector $\vecy$;
we have $C_\vecc(x)=C(x)$.

Let $f_k$ be a cheapest flow of size $k$, and let
$f_{k+1}$ be a cheapest flow of size $k+1$, both with respect to cost vector $\vecc$.
Let $n_k$ denote the number of edges in $f_k$.
Assume without loss of generality that the edges
in $f_k$ are labeled as $e_1, \dots, e_{n_k}$.

We will now gradually increase the costs of edges in $f_k$
so that the resulting cost vector $\vecy$ satisfies
certain conditions.
Specifically, we start with $\vecy=\vecc$.
Then, at $i$-th step, $i=1, \dots, n_k$,
we increase $\vecy(e_i)$ as much as possible
subject to the following constraints:

\begin{itemize}
\item[(a)] $\vecy(f_k)=\sum_{e\in f_k}\vecy(e)=C_\vecy(k)$, i.e., $f_k$ must remain the cheapest
      $k$-flow w.r.t. cost vector $\vecy$.
\item[(b)] $C_\vecy(k+1)-C_\vecy(k)=C(k+1)-C(k)$, i.e. $C_\vecy(k+1)-C_\vecy(k)$ does not
      change.
\end{itemize}

Since our $k$-path system is monopoly-free, in the end, all entries of $\vecy$
are finite. Further, it is not hard to see that when the process is over, we cannot
increase the cost of any edge in $f_k$ without violating~(a) or~(b).

Now, for each edge $e\in f_k$, we will define the {\em tight flow} $f(e)$ as below
to be a flow that prevented us from raising $\vecy(e)$ beyond its current value.
Specifically, consider each edge $e_i\in f_k$.
Suppose first that when we were raising $\vecy(e_i)$,
we had to stop because constraint~(a) became tight. In this case, let $f(e_i)$
be some cheapest flow of size $k$ in $G\setminus\{e_i\}$ with respect to the costs $\vecy$
at the end of stage $i$. Now, suppose that when we were raising $\vecy(e_i)$,
constraint~$(b)$ became tight first. In this case, let $f(e_i)$ be some cheapest
flow of size $k+1$ in $G\setminus\{e_i\}$ with respect to the costs $\vecy$
at the end of stage $i$. Observe that $f_k$ remains a cheapest $k$-flow
throughout the process; further, for all $e\in f_k$, the flow $f(e)$ is
a cheapest flow of its size in $G$ with respect to the final cost vector as well.
In the following we consider the cost vector $\vecy$ at the end of the process.

\begin{figure}[ht]
\begin{center}
\include{graphic-arxiv} 
\caption{The graph of $C_\vecy(x)$
}\label{graph}
\end{center}
\end{figure}

Let $f^*$ be the average of all tight flows, i.e., set
$$
f^*=\frac{1}{n_k}\sum_{e\in f_k}f(e).
$$
Let $q$ be the size of $f^*$; we have $k\le q\le k+1$. Note that
$f^*$ is a cheapest flow of size $q$ by the second statement of
lemma~\ref{lem:c(x)}, as it is a convex combination of cheapest flows of size
$k$ and cheapest flows of size $k+1$. Further, since $e\notin f(e)$
for any $e\in f_k$, the amount of flow that passes through each edge
$e$ in $f^*$ is strictly less than $1$. Thus, for a sufficiently
small $\epsilon>0$, flow $f^*+\epsilon f_k$ is a valid flow of size
$q+\epsilon k$ in $G$. Moreover, we have $C_\vecy(q+\varepsilon
k)\leq \vecy(f^{*}+\epsilon f_k) = C_\vecy(q)+\varepsilon
C_\vecy(k)$.

This observation, together with the convexity of $C_\vecy(x)$,
allows us to derive that $C_\vecy(x)$ is a linear function on the interval $[0,k+1]$.
Indeed, by convexity of $C_\vecy(x)$ we have
            \begin{eqnarray*}
            C_\vecy(k)\leq &\frac{q+(\varepsilon-1)k}{q+\varepsilon k}C_\vecy(0) +
            \frac{k}{q+\varepsilon k}C_\vecy(q+\varepsilon k)&=
            \frac{k}{q+\varepsilon k}C_\vecy(q+\varepsilon k) \\
            C_\vecy(q)\leq &\frac{\varepsilon k}{q+\varepsilon k}C_\vecy(0) +
            \frac{q}{q+\varepsilon k}C_\vecy(q+\varepsilon k)&=
            \frac{q}{q+\varepsilon k}C_\vecy(q+\varepsilon k)
            \end{eqnarray*}
If any of the two inequalities above is an equality, then $C_\vecy(x)$ is linear on
$[0,k+1]$ and we are done. Otherwise, both of these inequalities are strict,
and we can write
$$
C_\vecy(q+\varepsilon k)\leq C_\vecy(q)+\varepsilon C_\vecy(k)<
            \frac{q}{q+\varepsilon k}C_\vecy(q+\varepsilon k)+
            \varepsilon\frac{k}{q+\varepsilon k}C_\vecy(q+\varepsilon k)=
C_\vecy(q+\varepsilon k).
$$
The contradiction shows that $C_\vecy(x)$ is a linear function on $[0,k+1]$.

Since $\vecy$ satisfies conditions~(1) and~(2) in the definition of $\mu(\vecc)$,
we obtain
$$
\mu(\vecc)\geq C_\vecy(k)=k(C_\vecy(k+1)-C_\vecy(k))=k(C_\vecc(k+1)-C_\vecc(k)),
$$
where the first equality follows from linearity of $C_{\vecy}(x)$ and last equality holds by construction of $\vecy$.
Thus, the lemma is proven.
\end{proof}

Recall now that $L(G,c)$ by definition is the longest path in $G^*(\vecc)$,
where $G^*=G^*(\vecc)$ is the cheapest $k+1$-flow w.r.t. $\vecc$. Hence,
$C(k+1)=\vecc(G^*)$. Further we can decompose $G^*$ into the sum of
a path with the cost $L(G,c)$ and some $k$-flow, which implies that
$C(k+1)\ge L(G,c)+C(k)$. Rewriting the last inequality we get
$C(k+1)-C(k) \ge L(G,c)$. Combining this observation with
Lemma~\ref{lem:k-path-upper} we easily derive
Theorem~\ref{thm:alpha-mu}.

\section{Proof of Lemma~\ref{lem:vc-lowerbound}}
Let $C_v$ be some clique of size $\rho_v$ in $\CL_v$, and consider
the bid vector $\vecb$ given by $b(u)=x$ if $u\in C_v$ and $b(u)=0$
if $u\in V\setminus C_v$. Since $C_v$ is a clique, any vertex cover
for $G$ must contain at least $\rho_v-1$ vertices of $C_v$. Thus, any cheapest
feasible set with respect to the true costs contains all vertices in
$C_v\setminus\{v\}$; let $S$ denote some such set.
Moreover, for any $u\in C_v\setminus\{v\}$,
any vertex cover that does not contain $u$ must contain $v$, so
$\vecb$ satisfies condition (2) with respect to the set $S$
in the definition of the benchmark $\nu$.
To see that is also satisfies condition (3), note that if any vertex
in $C_v\setminus\{v\}$ decides to raise its bid, it can be replaced by its neighbors
at cost $x$. Now, consider any $w\in (V\setminus C_v)\cap S$.
The vertex $w$ cannot be adjacent
to all vertices in $C_v$, since otherwise $C_v$ would not be a maximal
clique. Thus, if $w\in S$, we can obtain a vertex cover of cost
$x(\rho_v-1)$ that does not include $w$ by taking all vertices
of cost $0$ as well as all vertices in $C_v$ that are adjacent to $w$.

%% file: 2paths-arxiv.tex
\begin{picture}(0,0)%
\includegraphics{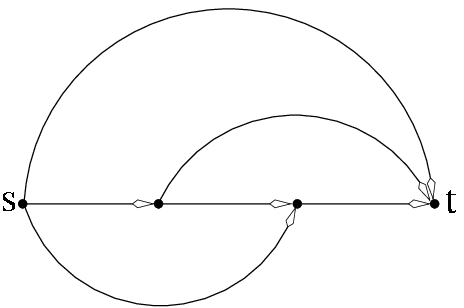}%
\end{picture}%
\setlength{\unitlength}{622sp}%
\begingroup\makeatletter\ifx\SetFigFontNFSS\undefined%
\gdef\SetFigFontNFSS#1#2#3#4#5{%
  \reset@font\fontsize{#1}{#2pt}%
  \fontfamily{#3}\fontseries{#4}\fontshape{#5}%
  \selectfont}%
\fi\endgroup%
\begin{picture}(13945,9139)(-31109,-2335)
\put(-28574,1019){\makebox(0,0)[lb]{\smash{{\SetFigFontNFSS{8}{9.6}{\rmdefault}{\mddefault}{\updefault}{\color[rgb]{0,0,0}$0$}%
}}}}
\put(-24434,6014){\makebox(0,0)[lb]{\smash{{\SetFigFontNFSS{8}{9.6}{\rmdefault}{\mddefault}{\updefault}{\color[rgb]{0,0,0}$3$}%
}}}}
\put(-22319,3674){\makebox(0,0)[lb]{\smash{{\SetFigFontNFSS{8}{9.6}{\rmdefault}{\mddefault}{\updefault}{\color[rgb]{0,0,0}$2$}%
}}}}
\put(-26504,-2086){\makebox(0,0)[lb]{\smash{{\SetFigFontNFSS{8}{9.6}{\rmdefault}{\mddefault}{\updefault}{\color[rgb]{0,0,0}$2$}%
}}}}
\put(-24659,1019){\makebox(0,0)[lb]{\smash{{\SetFigFontNFSS{8}{9.6}{\rmdefault}{\mddefault}{\updefault}{\color[rgb]{0,0,0}$0$}%
}}}}
\put(-20294,1019){\makebox(0,0)[lb]{\smash{{\SetFigFontNFSS{8}{9.6}{\rmdefault}{\mddefault}{\updefault}{\color[rgb]{0,0,0}$0$}%
}}}}
\put(-22319,1064){\makebox(0,0)[lb]{\smash{{\SetFigFontNFSS{12}{14.4}{\rmdefault}{\mddefault}{\updefault}{\color[rgb]{0,0,0}v}%
}}}}
\put(-26549,-16){\makebox(0,0)[lb]{\smash{{\SetFigFontNFSS{12}{14.4}{\rmdefault}{\mddefault}{\updefault}{\color[rgb]{0,0,0}u}%
}}}}
\end{picture}%

%% file: flow-arxiv.tex
\scalebox{0.6}{
\begin{picture}(-10,11)%
\includegraphics{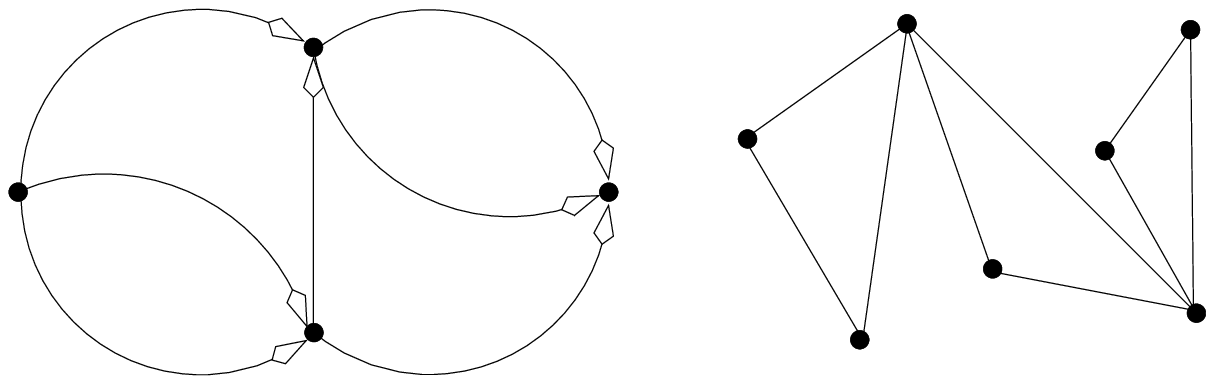}%
\end{picture}%
}
\setlength{\unitlength}{1243sp}%
\begingroup\makeatletter\ifx\SetFigFontNFSS\undefined%
\gdef\SetFigFontNFSS#1#2#3#4#5{%
  \reset@font\fontsize{#1}{#2pt}%
  \fontfamily{#3}\fontseries{#4}\fontshape{#5}%
  \selectfont}%
\fi\endgroup%
\scalebox{0.6}{
\begin{picture}(18795,7486)(-14324,-2420)
\put(-6164,3404){\makebox(0,0)[lb]{\smash{{\SetFigFontNFSS{14}{16.8}{\rmdefault}{\mddefault}{\updefault}{\color[rgb]{0,0,0}$e_5$}%
}}}}
\put(-5939,-1861){\makebox(0,0)[lb]{\smash{{\SetFigFontNFSS{14}{16.8}{\rmdefault}{\mddefault}{\updefault}{\color[rgb]{0,0,0}$e_7$}%
}}}}
\put(-13004,-1996){\makebox(0,0)[lb]{\smash{{\SetFigFontNFSS{14}{16.8}{\rmdefault}{\mddefault}{\updefault}{\color[rgb]{0,0,0}$e_3$}%
}}}}
\put(-9854,749){\makebox(0,0)[lb]{\smash{{\SetFigFontNFSS{14}{16.8}{\rmdefault}{\mddefault}{\updefault}{\color[rgb]{0,0,0}$e_4$}%
}}}}
\put(-11249,1019){\makebox(0,0)[lb]{\smash{{\SetFigFontNFSS{14}{16.8}{\rmdefault}{\mddefault}{\updefault}{\color[rgb]{0,0,0}$e_2$}%
}}}}
\put(-7514,794){\makebox(0,0)[lb]{\smash{{\SetFigFontNFSS{14}{16.8}{\rmdefault}{\mddefault}{\updefault}{\color[rgb]{0,0,0}$e_6$}%
}}}}
\put(-12644,3539){\makebox(0,0)[lb]{\smash{{\SetFigFontNFSS{14}{16.8}{\rmdefault}{\mddefault}{\updefault}{\color[rgb]{0,0,0}$e_1$}%
}}}}
\put(-359,3809){\makebox(0,0)[lb]{\smash{{\SetFigFontNFSS{14}{16.8}{\rmdefault}{\mddefault}{\updefault}{\color[rgb]{0,0,0}$e_1$}%
}}}}
\put(-2834,2144){\makebox(0,0)[lb]{\smash{{\SetFigFontNFSS{14}{16.8}{\rmdefault}{\mddefault}{\updefault}{\color[rgb]{0,0,0}$e_2$}%
}}}}
\put(586,-871){\makebox(0,0)[lb]{\smash{{\SetFigFontNFSS{14}{16.8}{\rmdefault}{\mddefault}{\updefault}{\color[rgb]{0,0,0}$e_4$}%
}}}}
\put(3871,3584){\makebox(0,0)[lb]{\smash{{\SetFigFontNFSS{14}{16.8}{\rmdefault}{\mddefault}{\updefault}{\color[rgb]{0,0,0}$e_5$}%
}}}}
\put(-1079,-2176){\makebox(0,0)[lb]{\smash{{\SetFigFontNFSS{14}{16.8}{\rmdefault}{\mddefault}{\updefault}{\color[rgb]{0,0,0}$e_3$}%
}}}}
\put(2476,1829){\makebox(0,0)[lb]{\smash{{\SetFigFontNFSS{14}{16.8}{\rmdefault}{\mddefault}{\updefault}{\color[rgb]{0,0,0}$e_6$}%
}}}}
\put(3961,-1681){\makebox(0,0)[lb]{\smash{{\SetFigFontNFSS{14}{16.8}{\rmdefault}{\mddefault}{\updefault}{\color[rgb]{0,0,0}$e_7$}%
}}}}
\put(-2609,4079){\makebox(0,0)[lb]{\smash{{\SetFigFontNFSS{20}{24.0}{\rmdefault}{\mddefault}{\updefault}{\color[rgb]{0,0,0}$\mathcal{H}$}%
}}}}
\put(-9719,4079){\makebox(0,0)[lb]{\smash{{\SetFigFontNFSS{20}{24.0}{\rmdefault}{\mddefault}{\updefault}{\color[rgb]{0,0,0}$G^*$}%
}}}}
\put(-4364,614){\makebox(0,0)[lb]{\smash{{\SetFigFontNFSS{14}{16.8}{\rmdefault}{\mddefault}{\updefault}{\color[rgb]{0,0,0}$t$}%
}}}}
\put(-14309,614){\makebox(0,0)[lb]{\smash{{\SetFigFontNFSS{14}{16.8}{\rmdefault}{\mddefault}{\updefault}{\color[rgb]{0,0,0}$s$}%
}}}}
\end{picture}%
} 

%% file: graphic-arxiv.tex
\begin{picture}(0,0)%
\includegraphics{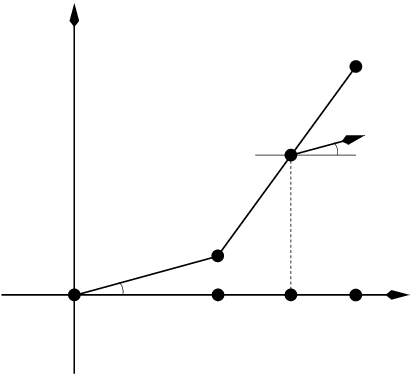}%
\end{picture}%
\setlength{\unitlength}{829sp}%
\begingroup\makeatletter\ifx\SetFigFontNFSS\undefined%
\gdef\SetFigFontNFSS#1#2#3#4#5{%
  \reset@font\fontsize{#1}{#2pt}%
  \fontfamily{#3}\fontseries{#4}\fontshape{#5}%
  \selectfont}%
\fi\endgroup%
\begin{picture}(9381,8526)(4288,-7924)
\put(13366,-5866){\makebox(0,0)[b]{\smash{{\SetFigFontNFSS{6}{7.2}{\rmdefault}{\mddefault}{\updefault}{\color[rgb]{0,0,0}$x$}%
}}}}
\put(6706,164){\makebox(0,0)[b]{\smash{{\SetFigFontNFSS{6}{7.2}{\rmdefault}{\mddefault}{\updefault}{\color[rgb]{0,0,0}$C(x)$}%
}}}}
\put(10936,-6676){\makebox(0,0)[b]{\smash{{\SetFigFontNFSS{6}{7.2}{\rmdefault}{\mddefault}{\updefault}{\color[rgb]{0,0,0}$q$}%
}}}}
\put(12421,-6676){\makebox(0,0)[b]{\smash{{\SetFigFontNFSS{6}{7.2}{\rmdefault}{\mddefault}{\updefault}{\color[rgb]{0,0,0}$k+1$}%
}}}}
\put(9271,-6721){\makebox(0,0)[b]{\smash{{\SetFigFontNFSS{6}{7.2}{\rmdefault}{\mddefault}{\updefault}{\color[rgb]{0,0,0}$k$}%
}}}}
\put(12736,-2356){\makebox(0,0)[rb]{\smash{{\SetFigFontNFSS{6}{7.2}{\rmdefault}{\mddefault}{\updefault}{\color[rgb]{0,0,0}$\varepsilon k$}%
}}}}
\end{picture}%